\documentclass[12pt]{article}
\usepackage{graphicx} 

\usepackage[margin=2cm]{geometry}

\usepackage{tikz}
\usepackage{amsmath}
\usepackage{amsfonts}
\usepackage{algorithm}
\usepackage{algpseudocode}
\usepackage{subcaption}
\usepackage{hyperref}
\usepackage{wrapfig}
\usepackage{mathtools}

\usepackage[font=bf]{caption}

\usepackage{natbib}
\usetikzlibrary{arrows,shapes,decorations.markings,arrows.meta,bending}
\usetikzlibrary{decorations.pathreplacing,angles,quotes}
\definecolor{primary}{RGB}{0, 99, 130}
\definecolor{secondary}{RGB}{180, 18, 27}

\usepackage[parfill]{parskip}

\newtheorem{theorem}{Theorem}%
\newtheorem{proposition}[theorem]{Proposition}%

\newcommand{\Pbb}{\mathbb{P}}
\newcommand{\thetab}{\boldsymbol{\theta}}
\newcommand{\Xb}{\boldsymbol{X}}
\newcommand{\Yb}{\boldsymbol{Y}}
\newcommand{\Ub}{\boldsymbol{U}}
\newcommand{\Nc}{\mathcal{N}}
\newcommand{\Tc}{\mathcal{T}}
\newcommand{\Sc}{\mathcal{S}}
\newcommand{\Kc}{\mathcal{K}}

\definecolor{tab_blue}{HTML}{1f77b4}
\definecolor{tab_red}{HTML}{d62728}
\definecolor{tab_green}{HTML}{2ca02c}
\definecolor{tab_olive}{HTML}{bcbd22}
\definecolor{tab_purple}{HTML}{9467bd}
\definecolor{tab_orange}{HTML}{ff7f0e}

\title{Non-centred Bayesian inference for discrete-valued state-transition models: the Rippler algorithm}
\author{James Neill\textsuperscript{1}, Lloyd A. C. Chapman\textsuperscript{1,$\dagger$} and Chris Jewell\textsuperscript{1,$\dagger$}}
\date{\footnotesize \textsuperscript{1}School of Mathematical Sciences, Lancaster University, Lancaster, LA1 4YF, UK \\ ~ \\ Corresponding author: \href{email:j.neill@lancaster.ac.uk}{j.neill@lancaster.ac.uk} \\ $\dagger$ Joint last authors.}

\begin{document}

\maketitle

\begin{abstract}
    Stochastic state-transition models of infectious disease transmission can be used to deduce relevant drivers of transmission when fitted to data using statistically principled methods. Fitting this individual-level data requires inference on individuals' unobserved disease statuses over time, which form a high-dimensional and highly correlated state space. We introduce a novel Bayesian (data-augmentation Markov chain Monte Carlo) algorithm for jointly estimating the model parameters and unobserved disease statuses, which we call the Rippler algorithm. This is a non-centred method that can be applied to any individual-based state-transition model. We compare the Rippler algorithm to the state-of-the-art inference methods for individual-based stochastic epidemic models and find that it performs better than these methods as the number of disease states in the model increases.
\end{abstract}


\section{Introduction}

Infectious disease transmission is often analysed using stochastic individual-based models. These epidemic models often comprise a state-transition process, where the population of interest is partitioned into a discrete set of states that individuals can transition between. Importantly, the rate of transition between states for one individual can depend on the current state of other individuals, complicating epidemic model inference.  

Hidden Markov models are a natural way of modelling infectious disease dynamics in discrete time; observation data are assumed to be noisy measurements of a hidden epidemic process. An individual-based (discrete-time) epidemic can be formulated as a coupled hidden Markov model (CHMM). A CHMM is a collection of hidden Markov models, with an added dependency structure between the hidden states in each Markov chain: the hidden state of any chain at a given time-point is dependent on the hidden states of \textit{all chains} at the previous time-point. The observations for any chain are still only dependent on the corresponding hidden states of that chain. Each chain represents one individual's true status over the course of the epidemic.

The hidden states in a CHMM occupy a high-dimensional and highly-correlated state space, making the use of frequentist methods for inference difficult. Instead, Bayesian methods are used to estimate the hidden states and model parameters jointly. The most commonly used technique is reversible-jump Markov chain Monte Carlo (MCMC) \citep{o1999bayesian, gibson1998estimating}. This method is computationally less demanding per update to the hidden state than other methods, but typically produces more highly correlated samples and requires a unique implementation for different epidemic models \citep{spencer2015super, chapman2020inferring}. Importance sampling methods are challenged as the dimension of the latent state increases \citep{rimella2023inference}. An alternative technique is individual forward filtering backward sampling (iFFBS), originally proposed in \cite{touloupou2020scalable}. This is an extension of the general forward filtering backward sampling (FFBS) algorithm \citep{carter1994gibbs, chib1996calculating}. iFFBS produces less correlated samples than reversible-jump MCMC, but has a high computational cost for epidemic models with many compartments.

In this paper, we develop a novel MCMC algorithm for inference on the hidden states of a CHMM. Our algorithm proposes a new latent variable state by making a small change to an individual’s hidden states, and then simulating forwards using the transmission model. This is a non-centred approach that moves event probability calculations from the Metropolis-Hastings accept-reject step into the proposal. We call this algorithm the \textit{Rippler method}, as the small change causes ``ripples'' through the latent variable space. The algorithm is an extension of the method proposed in \cite{neill2025non}, expanded from one simple model to any CHMM. Section \ref{sec_model} defines relevant notation and interprets individual-based epidemic models as CHMMs. The Rippler method is detailed in Section \ref{sec_rippler}, where we also explain how the algorithm can be extended to incorporate the observation data in the proposal step. In Section \ref{sec_comp} we detail existing inference methods, and in Section \ref{sec_sim} we compare them to the Rippler algorithm using simulation studies. We find that the Rippler algorithm compares favourably to the other methods, exploring the latent variable space more efficiently than reversible-jump MCMC and scaling better with the number of model compartments than iFFBS.

\section{Epidemic Modelling} \label{sec_model}

\subsection{Coupled Hidden Markov Models}

Consider a CHMM with $\Nc$ individuals (chains) and $\Tc$ time-points. Let $x_{t,j}$ be the hidden state for individual $j \in \{1,\dots,\Nc\}$ at time-point $t \in \{1,\dots,\Tc\}$. We assume there are a finite number of hidden states $\Sc$, so let $x_{t,j} \in \{1,\dots,\Sc\}$. Let $y_{t,j}$ be the observed state corresponding to $x_{t,j}$. We make no assumptions on the state space of $y_{t,j}$. Note that some $y_{t,j}$ may be undefined because of missing observation data. The hidden states are collected as $\Xb = (x_{t,j})_{t \in \{1,\dots,\Tc\}, j \in \{1,\dots,\Nc\}}$ and the observed states are collected as $\Yb = (y_{t,j})_{t \in \{1,\dots,\Tc\}, j \in \{1,\dots,\Nc\}}$. A visualisation of a CHMM is shown in Figure \ref{fig_chmm}.

Let $\thetab$ be the vector of model parameters. For each individual $j$ and time-point $t$, the value of the next hidden state $x_{t+1,j}$ is dependent on the model parameters $\thetab,$ the current hidden state $x_{t,j}$, and the current hidden states of all other individuals $\boldsymbol{x}_{t,-j} = (x_{t,i})_{i\in\{1,\dots,\Nc\}\backslash\{j\}}$. The hidden state of each individual is an inhomogeneous Markov chain with transition probabilities \[ p^{(t,j)}_{r,s} = \Pbb(x_{t+1,j}=s|x_{t,j}=r,\boldsymbol{x}_{t,-j},\thetab) \] for all $r,s \in \{1,\dots,\Sc\}$. The calculation of these transition probabilities for epidemic models is discussed in Section \ref{sec_model_sub_probs}. The value of the hidden state $x_{1,j}$ is determined by some prior distribution, dependent on the parameters $\thetab$. The value of each observation $y_{t,j}$ is dependent on the corresponding hidden state $x_{t,j}$ and the parameters $\thetab$, and so is related by some probability mass (or density) function $f(y_{t,j}|x_{t,j},\thetab)$. If $y_{t,j}$ is undefined then we have $f(y_{t,j}|x_{t,j},\thetab) = 1$.

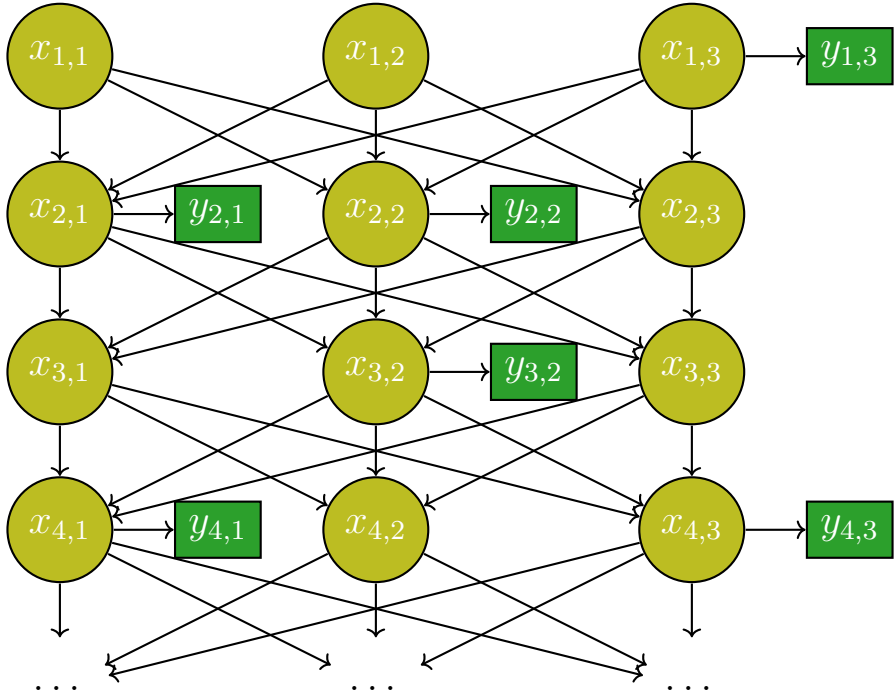
\begin{figure}[H]
    \centering
    \begin{tikzpicture}[thick, scale=1.4, every node/.style={scale=1.3}]
        \node[thick, draw, circle, minimum size=1cm, fill=tab_olive, text=white] (x10) at (0,0) {$x_{1,1}$};
        \node[thick, draw, circle, minimum size=1cm, fill=tab_olive, text=white] (x11) at (0,-1.5) {$x_{2,1}$};
        \node[thick, draw, rectangle, minimum size=0.1cm, fill=tab_green, text=white] (y11) at (1.5,-1.5) {$y_{2,1}$};
        \node[thick, draw, circle, minimum size=1cm, fill=tab_olive, text=white] (x12) at (0,-3) {$x_{3,1}$};
        \node[thick, draw, circle, minimum size=1cm, fill=tab_olive, text=white] (x13) at (0,-4.5) {$x_{4,1}$};
        \node[thick, draw, rectangle, minimum size=0.1cm, fill=tab_green, text=white] (y13) at (1.5,-4.5) {$y_{4,1}$};
        \node[thick, circle, minimum size=1cm] (dots1) at (0,-6) {$\dots$};
        \draw[->,thick] (x11)--(y11);
        \draw[->,thick] (x13)--(y13);
        \draw[->,thick] (x10)--(x11);
        \draw[->,thick] (x11)--(x12);
        \draw[->,thick] (x12)--(x13);
        \draw[->,thick] (x13)--(dots1);
        \node[thick, draw, circle, minimum size=1cm, fill=tab_olive, text=white] (x20) at (3,0) {$x_{1,2}$};
        \node[thick, draw, circle, minimum size=1cm, fill=tab_olive, text=white] (x21) at (3,-1.5) {$x_{2,2}$};
        \node[thick, draw, rectangle, minimum size=0.1cm, fill=tab_green, text=white] (y21) at (4.5,-1.5) {$y_{2,2}$};
        \node[thick, draw, circle, minimum size=1cm, fill=tab_olive, text=white] (x22) at (3,-3) {$x_{3,2}$};
        \node[thick, draw, rectangle, minimum size=0.1cm, fill=tab_green, text=white] (y22) at (4.5,-3) {$y_{3,2}$};
        \node[thick, draw, circle, minimum size=1cm, fill=tab_olive, text=white] (x23) at (3,-4.5) {$x_{4,2}$};
        \node[thick, circle, minimum size=1cm] (dots2) at (3,-6) {$\dots$};
        \draw[->,thick] (x21)--(y21);
        \draw[->,thick] (x22)--(y22);
        \draw[->,thick] (x20)--(x21);
        \draw[->,thick] (x21)--(x22);
        \draw[->,thick] (x22)--(x23);
        \draw[->,thick] (x23)--(dots2);
        \node[thick, draw, circle, minimum size=1cm, fill=tab_olive, text=white] (x30) at (6,0) {$x_{1,3}$};
        \node[thick, draw, rectangle, minimum size=0.1cm, fill=tab_green, text=white] (y30) at (7.5,0) {$y_{1,3}$};
        \node[thick, draw, circle, minimum size=1cm, fill=tab_olive, text=white] (x31) at (6,-1.5) {$x_{2,3}$};
        \node[thick, draw, circle, minimum size=1cm, fill=tab_olive, text=white] (x32) at (6,-3) {$x_{3,3}$};
        \node[thick, draw, circle, minimum size=1cm, fill=tab_olive, text=white] (x33) at (6,-4.5) {$x_{4,3}$};
        \node[thick, draw, rectangle, minimum size=0.1cm, fill=tab_green, text=white] (y33) at (7.5,-4.5) {$y_{4,3}$};
        \node[thick, circle, minimum size=1cm] (dots3) at (6,-6) {$\dots$};
        \draw[->,thick] (x30)--(y30);
        \draw[->,thick] (x33)--(y33);
        \draw[->,thick] (x30)--(x31);
        \draw[->,thick] (x31)--(x32);
        \draw[->,thick] (x32)--(x33);
        \draw[->,thick] (x33)--(dots3);
        \draw[->,thick] (x10)--(x21);
        \draw[->,thick] (x10)--(x31);
        \draw[->,thick] (x20)--(x11);
        \draw[->,thick] (x20)--(x31);
        \draw[->,thick] (x30)--(x11);
        \draw[->,thick] (x30)--(x21);
        \draw[->,thick] (x11)--(x22);
        \draw[->,thick] (x11)--(x32);
        \draw[->,thick] (x21)--(x12);
        \draw[->,thick] (x21)--(x32);
        \draw[->,thick] (x31)--(x12);
        \draw[->,thick] (x31)--(x22);
        \draw[->,thick] (x12)--(x23);
        \draw[->,thick] (x12)--(x33);
        \draw[->,thick] (x22)--(x13);
        \draw[->,thick] (x22)--(x33);
        \draw[->,thick] (x32)--(x13);
        \draw[->,thick] (x32)--(x23);
        \draw[->,thick] (x13)--(dots2);
        \draw[->,thick] (x13)--(dots3);
        \draw[->,thick] (x23)--(dots1);
        \draw[->,thick] (x23)--(dots3);
        \draw[->,thick] (x33)--(dots1);
        \draw[->,thick] (x33)--(dots2);
    \end{tikzpicture}
    \caption{Visualisation of a coupled hidden Markov model with $\Nc =3$ (from $t=1$ to $t=4$). Note that only a subset of time-individual points have observations. Based on a figure from \cite{touloupou2020scalable}.}
    \label{fig_chmm}
\end{figure}

\subsection{Transition Probabilities for Epidemics} \label{sec_model_sub_probs}

In epidemic models of interest, individuals are partitioned into a finite set of states, and can transition between those states at discrete time-points. The rate of transition between states is dependent on the current state of other individuals, the model parameters, and the current time-point. More formally, the behaviour of a given model (at time-point $t$ for individual $j$) can be expressed using a Markov transition rate matrix $\boldsymbol{Q}^{(t,j)} = (q^{(t,j)}_{r,s})_{r,s \in \{1,\dots,\Sc\}}$. If $r \neq s$, then $q^{(t,j)}_{r,s}$ is the rate of transition from state $r$ to state $s$. If $s = r$, then $q^{(t,j)}_{r,r}$ is the negative total rate of transition out of state $r$, and so $q^{(t,j)}_{r,r} = -\sum_{s=1, s\neq r}^\Sc q^{(t,j)}_{r,s}$.

We are interested in determining the transition probability matrix $\boldsymbol{P}^{(t,j)} = (p^{(t,j)}_{r,s})_{r,s \in \{1,\dots,\Sc\}}$. The direct approach to calculate $\boldsymbol{P}^{(t,j)}$ from $\boldsymbol{Q}^{(t,j)}$ is to use the matrix exponential $\boldsymbol{P}^{(t,j)}=e^{\boldsymbol{Q}^{(t,j)}}$. However, calculating the matrix exponential is often computationally slow, and so we use a different approach. The time until the next transition for each individual is assumed exponentially distributed, and so the probability of any transition occurring in one time unit is $1-e^{q^{(t,j)}_{r,r}}$, where $r$ is the current state of individual j at time-point $t$. The probability of transitioning to any particular state $s \neq r$ is then proportional to the ratio of the transition rates $q^{(t,j)}_{r,s}/-q^{(t,j)}_{r,r}$. This means \[ p^{(t,j)}_{r,s} = \begin{cases}
    0 & \text{ if } s \neq r, ~q^{(t,j)}_{r,r} = 0,
    \vspace{10pt}
    \\ \frac{q^{(t,j)}_{r,s}}{-q^{(t,j)}_{r,r}}\left(1-e^{q^{(t,j)}_{r,r}}\right) & \text{ if } s \neq r, ~q^{(t,j)}_{r,r} \neq 0,
    \vspace{10pt}
    \\ e^{q^{(t,j)}_{r,r}} & \text{ if } s = r,
\end{cases} \] for any $r,s \in \{1,\dots,\Sc\}$. 

\subsection{Simulation and Non-centering} \label{sec_model_sub_sim}

These transition probabilities can be used to simulate the hidden states $\Xb$ in the CHMM. If we only have two states, we draw from a Bernoulli distribution. In the general case of $\Sc$ states, we use a categorial distribution with $\Sc$ categories. For the initial hidden states, we draw \[ x_{1,j} \sim \text{Categorial}(\Sc,(\tilde{p}^{(j)}_{1},\dots,\tilde{p}^{(j)}_{\Sc})) \] for every individual $j \in \{1,\dots,\Nc\}$, where $\tilde{p}^{(j)}_s$ is the prior probability of individual $j$ initially being in state $s$. Then for $t \in \{1,\dots,\Tc-1\}$ we draw \[ x_{t+1,j} \sim \text{Categorial}(\Sc,(p^{(t,j)}_{x_{t,j},1},\dots,p^{(t,j)}_{x_{t,j},\Sc})) \] for every individual $j \in \{1,\dots,\Nc\}$. These random draws can be simulated in parallel \textit{at any given time-point}, but must be simulated in order \textit{across time-points}.

Alternatively, we can take a non-centred approach to determining $\Xb$. Instead of drawing directly from a Categorical distribution, we draw $u \sim \text{Unif}(0,1)$ and determine the hidden state by comparing $u$ to the transition probabilities. Let $g$ be the function that compares $u$ to the transition probabilities $(p_1,\dots,p_\Sc)$; $g(u,\Sc,(p_1,\dots,p_\Sc))$ returns $s \in \{1,\dots,\Sc\}$ such that $\sum_{r=1}^{s-1} p_r < u < \sum_{r=1}^{s} p_r$. Details of the derivation of $g$ are given in Section \ref{sec_supp_noncentred}. For the initial hidden states, we draw $u_{1,j} \sim \text{Unif}(0,1)$ and calculate \begin{equation*}
    x_{1,j} = g(u_{1,j},\Sc,(\tilde{p}^{(j)}_{1},\dots,\tilde{p}^{(j)}_{\Sc})) 
\end{equation*} for every individual $j \in \{1,\dots,\Nc\}$. Then for $t \in \{1,\dots,\Tc-1\}$ we draw $u_{t+1,j} \sim \text{Unif}(0,1)$ and calculate \begin{equation*}
    x_{t+1,j} = g(u_{t+1,j},\Sc,(p^{(t,j)}_{x_{t,j},1},\dots,p^{(t,j)}_{x_{t,j},\Sc})) 
\end{equation*} for every individual $j \in \{1,\dots,\Nc\}$. Note that while all random draws can be simulated in parallel, we must still make the calculations in order across time-points (since determining $x_{t+1,j}$ requires the value of $x_{t,j}$). We see that all elements of $\Ub = (u_{t,j})_{t \in \{1,\dots,\Tc\}, j \in \{1,\dots,\Nc\}}$ are \textit{a priori} i.i.d. -- we will use $\Ub$ as a re-parameterisation of $\Xb$ in Section \ref{sec_rippler}. 

\section{Bayesian Inference: The Rippler Algorithm} \label{sec_rippler}

\subsection{Overview} 

We now introduce the Rippler algorithm: a data-augmentation MCMC algorithm to perform inference for the transmission states $\Xb$ (and parameters $\thetab$), given partial test information $\Yb$, i.e.~to infer the joint posterior distribution $\pi(\Xb,\thetab|\Yb)$. We use a Metropolis-within-Gibbs algorithm to update $\thetab$ and $\Xb$ separately, alternating between draws from $\pi(\thetab|\Xb,\Yb)$ and $\pi(\Xb|\thetab,\Yb)$.

The Rippler algorithm takes a non-centred approach to latent variable updates. We materialise a matrix of random numbers $\Ub$ that, along with the parameters $\thetab$, could have produced the current latent variables $\Xb$. Then we make a change to $\Ub$ in some way; our default method is to change one element of the matrix and leave all the others the same, but we also explore other options. This change provides us with the matrix $\Ub^\ast$. We then convert back to the latent variables, giving us $\Xb^{\ast}$. Finally, we accept or reject $\Xb^\ast$ using a Metropolis-Hastings accept-reject step.

There are many known methods to update the parameters $\thetab$ from $\Xb$ and $\Yb$, such as random walk Metropolis, the Metropolis-adjusted Langevin algorithm, and Hamiltonian Monte Carlo. Here we focus on the latent variable update; the choice of the parameter update algorithm is left to the reader.

\subsection{Proposal Mechanism} \label{sec_rippler_sub_proposal}

The first step of the proposal is to materialise a matrix of random numbers $\Ub$ that, when combined with $\thetab$, could produce $\Xb$. To do this we will calculate the lower and upper bounds on the values of random numbers that could produce the latent variables; we collect these values in matrices $\Ub^{low} = (u^{low}_{t,j})_{t \in \{1,\dots,\Tc\}, j \in \{1,\dots,\Nc\}}$ and $\Ub^{upp} = (u^{upp}_{t,j})_{t \in \{1,\dots,\Tc\}, j \in \{1,\dots,\Nc\}}$ respectively. Full details are given in Section \ref{sec_supp_rippler_sub_rippler}. Note that since we already know the latent variables these calculations can be performed in parallel.

We then materialise each random number from the largest possible range of values that produce the current latent variables; we let $\Ub = (u_{t,j})_{t \in \{1,\dots,\Tc\}, j \in \{1,\dots,\Nc\}}$, where $u_{t,j} \sim \text{Unif}(u^{low}_{t,j},u^{upp}_{t,j})$.

Next we to make some change to $\Ub$ in order to propose $\Ub^\ast$. We do this by randomly choosing a pair $(t,j)$, changing $u_{t,j}^\ast$, and (in our default approach) keeping all other elements of $\Ub^\ast$ the same as in $\Ub$.

The range of random numbers that could produce $x_{t,j}$ is $(u^{low}_{t,j},u^{upp}_{t,j})$. If we propose $u^\ast_{t,j}$ from this range, we would have $x^\ast_{t,j}=x_{t,j}$, and so $\Xb^\ast=\Xb$. To avoid this `null move', we instead propose $u^\ast_{t,j}$ uniformly from $(0,u^{low}_{t,j}) \cup (u^{upp}_{t,j},1)$. Each pair $(t,j)$ is chosen with probability proportional to the length of the region that $u^\ast_{t,j}$ can be proposed from: $1 - u^{upp}_{t,j} + u^{low}_{t,j}$.

The last part of the proposal is to convert $\Ub^\ast$ back into latent variables $\Xb^{\ast}$. We do this using the non-centred simulation approach described in Section \ref{sec_model_sub_sim}, using the random numbers in $\Ub^{\ast}$ instead of drawing new random numbers. For the initial hidden states, we calculate \begin{equation}
    x^\ast_{1,j} = g(u^\ast_{1,j},\Sc,(\tilde{p}^{(j)}_{1},\dots,\tilde{p}^{(j)}_{\Sc})), \label{x_initial_func}
\end{equation} for every individual $j \in \{1,\dots,\Nc\}$. Then for $t \in \{1,\dots,\Tc-1\}$ we calculate \begin{equation}
    x^\ast_{t+1,j} = g(u^\ast_{t+1,j},\Sc,(p^{(t,j)}_{x^\ast_{t,j},1},\dots,p^{(t,j)}_{x^\ast_{t,j},\Sc})) \label{x_full_func}
\end{equation} for every individual $j \in \{1,\dots,\Nc\}$. 

Using $\Ub^{\ast}$ rather than $\Ub$ to generate the latent variables proposes new transmission states $\Xb^\ast$. In all proposed $\Xb^\ast$ we have $x^\ast_{t,j} \neq x_{t,j}$. However, by changing $x_{t,j}$ to $x^\ast_{t,j}$, the transition probabilities at time $t+1$ may be different under $\Xb^\ast$ than under $\Xb$. This could change $\boldsymbol{x}_{t+1}$, which means the transition probabilities at time $t+2$ may be different, and so on (the single change ``ripples'' through the latent variables). We perform a Metropolis-Hastings accept-reject step to determine if we will accept the proposed latent variables $\Xb^\ast$.

\subsection{Acceptance Probability} \label{sec_rippler_sub_acceptance}

The new latent variables $\Xb^\ast$ are accepted with probability \[\alpha(\Xb,\Xb^{\ast}) = \min\left\{1,\frac{\pi(\Xb^{\ast}|\boldsymbol{\thetab,\Yb})}{\pi(\Xb|\boldsymbol{\thetab,\Yb})}\frac{q(\Xb|\thetab,\Xb^{\ast})}{q(\Xb^{\ast}|\thetab,\Xb)}\right\}.\] Since we know $\pi(\Xb|\thetab,\Yb) \propto \pi(\thetab,\Xb|\Yb)$, we use Bayes' theorem to split the acceptance rate into its component terms: \[\alpha(\Xb,\Xb^{\ast}) = \min\left\{1,\frac{\pi(\Yb|\thetab,\Xb^\ast)\pi(\Xb^{\ast}|\thetab)}{\pi(\Yb|\thetab,\Xb)\pi(\Xb|\thetab)}\frac{q_1(\Ub^{\ast}|\thetab,\Xb^{\ast})q_2(\Ub|\thetab,\Ub^{\ast},\Xb^\ast)q_3(\Xb|\thetab,\Ub)}{q_1(\Ub|\thetab,\Xb)q_2(\Ub^{\ast}|\thetab,\Ub,\Xb)q_3(\Xb^{\ast}|\thetab,\Ub^{\ast})}\right\}.\] We simplify this expression using the following propositions.

\begin{proposition}
    We have $\pi(\Xb|\thetab)q_1(\Ub|\thetab,\Xb) = 1$ and $\pi(\Xb^{\ast}|\thetab)q_1(\Ub^{\ast}|\thetab,\Xb^{\ast}) = 1$. \label{prop1}
\end{proposition}

\textit{Proof.} See Section \ref{sec_supp_rippler_sub_rippler}.

\begin{proposition}
    We have $q_3(\Xb^{\ast}|\thetab,\Ub^\ast) = 1$ and $q_3(\Xb|\thetab,\Ub) = 1$. \label{prop2}
\end{proposition}

\textit{Proof.} We propose $\Xb^\ast$ deterministically from $\Ub^\ast$ and $\thetab$. This means that any given $\Ub^\ast$ will always correspond to the same $\Xb^\ast$, and so $q_3(\Xb^{\ast}|\thetab,\Ub^\ast) = 1$. By symmetry, the same result follows for the reverse move. \begin{flushright} $\Box$ \end{flushright}

This means the acceptance rate simplifies to \[\alpha(\Xb,\Xb^{\ast}) = \min\left\{1,\frac{\pi(\Yb|\thetab,\Xb^\ast)}{\pi(\Yb|\thetab,\Xb)}\frac{q_2(\Ub|\thetab,\Ub^{\ast},\Xb^\ast)}{q_2(\Ub^{\ast}|\thetab,\Ub,\Xb)}\right\},\] where \[ \pi(\Yb|\thetab,\Xb^\ast) = \prod_{t=1}^\Tc \prod_{j=1}^\Nc f(y_{t,j}|x_{t,j}^\ast,\thetab),\] and \[ q_2(\Ub^{\ast}|\thetab,\Ub,\Xb) = \frac{1}{\sum_{t=1}^\Tc\sum_{j=1}^\Nc(1 - u^{upp}_{t,j} + u^{low}_{t,j})}.\]

\subsection{Adaptive Tuning} \label{sec_rippler_sub_tuning}

In the second step of the Rippler proposal we make some change to $\Ub$ in order to produce $\Ub^\ast$. By default, we change one element of $\Ub$ and keep all other elements of $\Ub^\ast$ the same as in $\Ub$ (see Section \ref{sec_rippler_sub_proposal}). We can, however, change an arbitrary $\kappa$ elements of $\Ub$, where $\kappa$ is some positive integer.

We wish to set $\kappa$ to maximise the mixing of the Markov chain. However, as we do not know the optimal value before running the algorithm, we take an adaptive approach. Let $a_\kappa$ be the acceptance rate of Rippler updates when changing $\kappa$ elements of $\Ub$. Our adaptive tuning algorithm takes an $\epsilon$-greedy approach: we take exploration steps (choosing a random $\kappa$ up to some $\kappa_{max}$) with some probability $\epsilon$, and take exploitation steps (choosing $\kappa$ such that $a_\kappa$ is as close as possible to a target acceptance rate $a'$) with probability $1-\epsilon$.

More formally, with probability $\epsilon$, we choose $\kappa$ uniformly from $\{1,\dots,\kappa_{max}\}$. Otherwise, let \[ \kappa = \arg \min \{ |a_{\kappa^\ast}-a'| : \kappa^\ast \in \{1,\dots,\kappa_{max}\}\}, \] where $a'$ is the target acceptance rate. We choose $a'=0.234$ as suggested in \cite{lee2018optimal}, a target which also closely matches our empirical results (see Section \ref{sec_supp_sim_sub_SIR}).


\subsection{Summary}

The full algorithm is presented below in Algorithm \ref{alg_rippler}. We alternate between parameter and latent variable updates. Since the latent variable space is so large, we update the latent variables multiple times per iteration.

\begin{algorithm}[H]
\caption{Rippler Inference}
\label{alg_rippler}
\begin{algorithmic}[1]
\Require Number of iterations $\Kc$, number of latent updates per iteration $\Kc'$, starting parameters $\thetab^{(0)}$, starting latent variables $\Xb^{(0)}$, observation data $\Yb$.
\State Let $\thetab=\thetab^{(0)}$ and $\Xb=\Xb^{(0)}$.
\For{$k \in \{1,\dots,\Kc\}$}
    \State Update $\thetab \sim \pi(\thetab|\Xb,\Yb)$ using a suitable MCMC update.
    \For{$k' \in \{1,\dots,\Kc'\}$}
        \State Calculate $\Ub^{low}$ and $\Ub^{upp}$ from $\thetab$ and $\Xb$ using equations \ref{u_bounds_initial} and \ref{u_bounds_full}.
        \State Let $\Ub = (u_{t,j})_{t \in \{1,\dots,\Tc\}, j \in \{1,\dots,\Nc\}}$, where $u_{t,j} \sim \text{Unif}(u^{low}_{t,j},u^{upp}_{t,j})$.
        \State Propose $\Ub^{\ast}$ from $\Ub$, $\Ub^{low}$, and $\Ub^{upp}$ (using adaptive tuning).
        \State Calculate $\Xb^{\ast}$ from $\thetab$ and $\Ub^{\ast}$ using equations \ref{x_initial_func} and \ref{x_full_func}.
        \State Calculate \[\alpha(\Xb,\Xb^{\ast}) = \min\left\{1,\frac{\pi(\Yb|\thetab,\Xb^\ast)}{\pi(\Yb|\thetab,\Xb)}\frac{q_2(\Ub|\thetab,\Ub^{\ast},\Xb^\ast)}{q_2(\Ub^{\ast}|\thetab,\Ub,\Xb)}\right\}.\]
        \State Let $\Xb = \Xb^{\ast}$ with probability $\alpha(\Xb,\Xb^{\ast})$.
    \EndFor
    \State Set $\thetab^{(k)}=\thetab$ and $\Xb^{(k)}=\Xb$.
\EndFor
\State \Return Parameter samples $\thetab^{(0)},\thetab^{(1)},\dots,\thetab^{(K)}$, latent variable samples $\Xb^{(0)},\Xb^{(1)},\dots,\Xb^{(K)}$.
\end{algorithmic}
\end{algorithm}

\subsection{Data-informed Rippler Algorithm} \label{sec_rippler_sub_hybrid}

In the standard Rippler algorithm, we use both initial state probabilities $ \tilde{p}^{(j)}_{s} = \Pbb(x_{1,j}=s|\thetab)$ and transition probabilities $p^{(t,j)}_{r,s} = \Pbb(x_{t+1,j}=s|x_{t,j}=r,\boldsymbol{x}_{t,-j},\thetab)$ to determine $\Ub^{low}$ and $\Ub^{upp}$ from $\Xb$, and to determine $\Xb^\ast$ from $\Ub^\ast$. We now present an alternative version of the Rippler algorithm that uses the observation data in the proposal step; we call this the data-informed Rippler algorithm. In this method, we use modified probabilities $\tilde{\rho}^{(j)}_{s} = \Pbb(x_{1,j}=s|y_{1,j},\thetab)$ and $\rho^{(t,j)}_{r,s} = \Pbb(x_{t+1,j}=s|x_{t,j}=r,\boldsymbol{x}_{t,-j},y_{t+1,j},\thetab)$ to replace $\tilde{p}^{(j)}_s$ and $p^{(t,j)}_{r,s}$ respectively. Bayes' theorem is used to calculate the modifications to the probabilities; we have \[ \tilde{\rho}^{(j)}_{s} = \frac{\tilde{p}^{(j)}_{s}f(y_{1,j}|s,\thetab)}{\sum_{s'=1}^\Sc \tilde{p}^{(j)}_{s'}f(y_{1,j}|s',\thetab)} \] and \[ \rho^{(t,j)}_{r,s} = \frac{p^{(t,j)}_{r,s}f(y_{t+1,j}|s,\thetab)}{\sum_{s'=1}^\Sc p^{(t,j)}_{r,s'}f(y_{t+1,j}|s',\thetab)}. \] The full calculations are presented in Section \ref{sec_supp_rippler_sub_hybrid}. When determining $\Xb^\ast$ from $\Ub^\ast$, for the initial hidden states we calculate \begin{equation}
    x^\ast_{1,j} = g(u^\ast_{1,j},\Sc,(\tilde{\rho}^{(j)}_{1},\dots,\tilde{\rho}^{(j)}_{\Sc})) \label{x_initial_func_hybrid}
\end{equation} for every individual $j \in \{1,\dots,\Nc\}$. Then for $t \in \{1,\dots,\Tc-1\}$ we calculate \begin{equation}
    x^\ast_{t+1,j} = g(u^\ast_{t+1,j},\Sc,(\rho^{(t,j)}_{x^\ast_{t,j},1},\dots,\rho^{(t,j)}_{x^\ast_{t,j},\Sc})) \label{x_full_func_hybrid}
\end{equation} for every individual $j \in \{1,\dots,\Nc\}$. The modifications to determining $\Ub^{low}$ and $\Ub^{upp}$ from $\Xb$ are detailed in Section \ref{sec_supp_rippler_sub_hybrid}.

Using the modified initial and transmission probabilities also changes the acceptance probability in the Metropolis-Hastings accept-reject step. Let $\tilde{c}^{(j)} = \sum_{s'=1}^\Sc \tilde{p}^{(j)}_{s'}f(y_{1,j}|s',\thetab)$ and $c^{(t,j)}_r = \sum_{s'=1}^\Sc p^{(t,j)}_{r,s'}f(y_{t+1,j}|s',\thetab)$ be the normalising constants when modifying the transmission probabilities. Note that the data-informed Rippler algorithm relies on an assumption that we always have $\tilde{c}^{(j)} \neq 0$ and $c^{(t,j)}_r \neq 0$.  We simplify the acceptance probability using the following proposition (and Proposition \ref{prop2}).

\begin{proposition}
    We have \[ \pi(\boldsymbol{Y|\thetab,\Xb})\pi(\Xb|\thetab)q_1(\Ub|\thetab,\Xb) = \prod_{j=1}^\Nc\left(\tilde{c}^{(j)}\prod_{t=1}^{\Tc-1}c^{(t,j)}_{x_{t,j}}\right) \] and \[ \pi(\boldsymbol{Y|\thetab,\Xb^\ast})\pi(\Xb^{\ast}|\thetab)q_1(\Ub^{\ast}|\thetab,\Xb^{\ast}) = \prod_{j=1}^\Nc\left(\tilde{c}^{(j)}\prod_{t=1}^{\Tc-1}c^{(t,j)}_{x^\ast_{t,j}}\right). \] \label{prop3}
\end{proposition}

\textit{Proof.} See Section \ref{sec_supp_rippler_sub_hybrid}.

This means the acceptance rate simplifies to \[ \alpha(\Xb,\Xb^{\ast}) = \min\left\{1,\frac{\prod_{j=1}^\Nc\prod_{t=1}^{\Tc-1}c^{(t,j)}_{x_{t,j}}}{\prod_{j=1}^\Nc\prod_{t=1}^{\Tc-1}c^{(t,j)}_{x^\ast_{t,j}}}\frac{q_2(\Ub|\thetab,\Ub^{\ast},\Xb^\ast)}{q_2(\Ub^{\ast}|\thetab,\Ub,\Xb)}\right\}. \] Note that we have made no changes to the proposal of $\Ub^\ast$ from $\Ub$ step, thus $q_2(\Ub^{\ast}|\thetab,\Ub,\Xb)$ and $q_2(\Ub|\thetab,\Ub^\ast,\Xb^\ast)$ are unchanged. The full algorithm is presented below in Algorithm \ref{alg_rippler_data_informed}.

\begin{algorithm}[H]
\caption{Data-informed Rippler Inference}
\label{alg_rippler_data_informed}
\begin{algorithmic}[1]
\Require Number of iterations $\Kc$, number of latent updates per iteration $\Kc'$, starting parameters $\thetab^{(0)}$, starting latent variables $\Xb^{(0)}$, observation data $\Yb$.
\State Let $\thetab=\thetab^{(0)}$ and $\Xb=\Xb^{(0)}$.
\For{$k \in \{1,\dots,\Kc\}$}
    \State Update $\thetab \sim \pi(\thetab|\Xb,\Yb)$ using a suitable MCMC update.
    \For{$k' \in \{1,\dots,\Kc'\}$}
        \State Calculate $\Ub^{low}$ and $\Ub^{upp}$ from $\thetab$ and $\Xb$ using equations \ref{u_bounds_initial_hybrid} and \ref{u_bounds_full_hybrid}.
        \State Let $\Ub = (u_{t,j})_{t \in \{1,\dots,\Tc\}, j \in \{1,\dots,\Nc\}}$, where $u_{t,j} \sim \text{Unif}(u^{low}_{t,j},u^{upp}_{t,j})$.
        \State Propose $\Ub^{\ast}$ from $\Ub$, $\Ub^{low}$, and $\Ub^{upp}$ (using adaptive tuning).
        \State Calculate $\Xb^{\ast}$ from $\thetab$ and $\Ub^{\ast}$ using equations \ref{x_initial_func_hybrid} and \ref{x_full_func_hybrid}.
        \State Calculate \[ \alpha(\Xb,\Xb^{\ast}) = \min\left\{1,\frac{\prod_{j=1}^\Nc\prod_{t=1}^{\Tc-1}c^{(t,j)}_{x_{t,j}}}{\prod_{j=1}^\Nc\prod_{t=1}^{\Tc-1}c^{(t,j)}_{x^\ast_{t,j}}}\frac{q_2(\Ub|\thetab,\Ub^{\ast},\Xb^\ast)}{q_2(\Ub^{\ast}|\thetab,\Ub,\Xb)}\right\}. \]
        \State Let $\Xb = \Xb^{\ast}$ with probability $\alpha(\Xb,\Xb^{\ast})$.
    \EndFor
    \State Set $\thetab^{(k)}=\thetab$ and $\Xb^{(k)}=\Xb$.
\EndFor
\State \Return Parameter samples $\thetab^{(0)},\thetab^{(1)},\dots,\thetab^{(K)}$, latent variable samples $\Xb^{(0)},\Xb^{(1)},\dots,\Xb^{(K)}$.
\end{algorithmic}
\end{algorithm}

\section{Bayesian Inference: Existing Methods} \label{sec_comp}



\subsection{Existing Inference Methods}

There are multiple existing methods for Bayesian inference on models of individual-based infectious disease transmission. The most common method for this is reversible-jump MCMC (RJMCMC), originally proposed in \cite{gibson1998estimating} and \cite{o1999bayesian}. Each RJMCMC update chooses an individual at random and proposes a change to their latent variables: either moving the time of a transition event, adding a transition event, or removing a transition event. The proposed change is then accepted or reject in a Metropolis-Hastings step. 

This method deals well with so-called `linear' models, where individuals can only ever transition into one other state and eventually reach a final absorbing state (e.g., the SIR model). However, many models are not so simple, and may contain at least one of two complicating factors: loops and branches. In a model with loops, individuals can re-enter states they have previous exited, meaning that state transitions can occur multiple times within the course of the epidemic \citep{spencer2015super}. In a model with branches, individuals may have more than one possible state they can move to when exiting a state \citep{chapman2020inferring}. RJMCMC struggles with these models: the latent variables (epidemic space) mix slowly and the algorithm can become very complicated/difficult to implement. By contrast, loops and branches do not pose an issue to the Rippler algorithm (which is independent of the model structure).

Alternatively, the state of the art for inference on coupled hidden Markov models is individual forward filtering backwards sampling (iFFBS), originally proposed in \cite{touloupou2020scalable}. This is an extension of the general forward filtering backwards sampling (FFBS) algorithm developed in \cite{carter1994gibbs} and \cite{chib1996calculating}. The FFBS algorithm first consists of iteratively calculating the probabilities of all possible latent variables at each time-point (given only the observations up to that time-point). After this, we use these probabilities to sample iteratively but backwards in time. This results in a sample from the full conditional distribution of the latent variables $\Xb$ given the parameters $\thetab$ and observations $\Yb$. However, in an individual-based model, the latent variables have $\Sc^\Nc$ states at any given time-point, and thus FFBS is computationally infeasible.

The iFFBS algorithm avoids this problem by choosing an individual at random and performing a FFBS update on that individual only -- the latent variables only have $S$ states at any given time-point. For all other individuals their latent variables are kept constant, with the forward filtering probabilities in the iFFBS update being modified from FFBS to reflect this. As such we expect the iFFBS algorithm to perform well at identifying within-individual correlation, but weaker at identifying between-individual correlation (where we believe the Rippler algorithm represents an improvement).


\subsection{Computational Complexity Comparison} \label{sec_comp_sub_scaling}

Direct comparison between the Rippler algorithm, RJMCMC, and iFFBS is difficult. One simple definition for the ``best'' algorithm is measuring how well the latent variable Markov chains mix per unit time of computation. However, even this seemingly basic definition proves complex, as there is no standard for the mixing of latent variable Markov chains for coupled hidden Markov models. One option is the autocorrelation function of some summary statistic \citep{touloupou2020scalable}, but the choice of this statistic may not be clear, or there may not be a summary statistic that adequately captures all latent variables. Another option is measuring the distance between iterations for each individual-time-point (e.g., mean absolute jump distance, root mean squared jump distance), but defining the distance between states is not always straightforward.

We avoid direct comparison of the time taken to run the algorithms, since this is affected by the code implementation. Instead we investigate the \textit{computational complexity scaling} as the number of states in the model $\Sc$, number of individuals $\Nc$, and number of time-points $\Tc$ change. Note that in \cite{touloupou2020scalable} the iFFBS algorithm performs $\Nc$ latent variable updates per parameter update (one per individual) -- we change this an arbitrary number of latent variable updates per parameter update, choosing individuals to update randomly. This ensures a fair comparison with the RJMCMC and Rippler algorithms, which also perform latent variable updates a fixed number of times per parameter update (enough to ensure sufficient mixing).

In the accept-reject step of an RJMCMC update, we calculate the probability of moving from state $x_{t-1,j}$ to $x_{t,j}$ for each $t \in \{2,\dots,\Tc\}$ and $j \in \{1,\dots,\Nc\}$. Therefore the computational complexity of each RJMCMC update is $O(\Nc\Tc)$. During the proposal step of a Rippler update (moving from random numbers back to the latent variables), we calculate the probability of moving from state $x_{t-1,j}$ to every $s \in \{1,\dots,\Sc\}$ for each $t \in \{2,\dots,\Tc\}$ and $j \in \{1,\dots,\Nc\}$. Therefore the computational complexity of each Rippler update is $O(\Sc\Nc\Tc)$. When calculating the forward probabilities in an iFFBS update (for some individual $i$), we calculate the probability of moving from state $r \in \{1,\dots,\Sc\}$ to $s \in \{1,\dots,\Sc\}$ for each $t \in \{2,\dots,\Tc\}$, and we modify these probabilities with an additional product for each $j \in \{1,\dots,\Nc\}\backslash \{i\}$. Therefore the computational complexity of each iFFBS update is $O(\Sc^2\Nc\Tc)$. 

\section{Simulation Studies} \label{sec_sim}

\subsection{Overview}

In this section, we conduct multiple simulation studies to assess the performance of the Rippler algorithms for updating the transmission states, and compare it with the performance of existing methods. In Section \ref{sec_sim_sub_SIR}, we test the convergence of the methods on the general stochastic SIR epidemic. In Sections \ref{sec_sim_sub_SEIR} and \ref{sec_sim_sub_SIS}, we test the methods on more complicated models and compare how they perform as the number of states changes. Additional tests, such as how the Rippler algorithm performs as we alter the number of random number changes in each update, are performed in Section \ref{sec_supp_sim}.

\subsection{SIR Model} \label{sec_sim_sub_SIR}

The first simulation study uses the SIR model (the general stochastic epidemic \citep{held2019handbook}). This model partitions the population of interest into three mutually-exclusive states: susceptive, infective, and recovered. Consider a population of $\Nc$ individuals over $\Tc$ time-points. Let $x_{t,j} \in \{1,2,3\}$ be the current state of individual $j \in \{1,\dots,\Nc\}$ at time-point $t \in \{1,\dots,\Tc\}$, where $x_{t,j}=1$ if they are susceptible, $x_{t,j}=2$ if they are infective, and $x_{t,j}=3$ if they are recovered. State transitions occur from susceptible to infective at rate $\beta$ for each currently infective individual, and from infective to recovered at constant rate $\gamma$. More specifically, we have transition probabilities \[ \boldsymbol{P}^{(t,j)} = \begin{pmatrix}
    e^{-\beta I(t)} & 1-e^{-\beta I(t)} & 0
    \\ 0 & e^{-\gamma} & 1-e^{-\gamma}
    \\ 0 & 0 & 1
\end{pmatrix}\] for any $t \in \{1,\dots,\Tc-1\}$ and $j \in \{1,\dots,\Nc\}$, where $I(t) = \sum_{j=1}^\Nc \mathbf{1}\{x_{t,j}=2\}$ (the total number of infective individuals at time $t$). A visualisation of the model is shown in Figure \ref{fig_sir_visual}.

\begin{figure}[H]
    \centering
    \begin{tikzpicture}[thick, scale=1.4, every node/.style={scale=1.3}]
        \node[thick, draw, rectangle, minimum size=1cm, fill=tab_blue, text=white] (S) at (0,0) {S};
        \node[thick, draw, rectangle, minimum size=1cm, fill=tab_red, text=white] (I) at (3,0) {I};
        \node[thick, draw, rectangle, minimum size=1cm, fill=tab_green, text=white] (R) at (6,0) {R};
        \draw[->,thick] (S)--(I);
        \draw[->,thick] (I)--(R);
        \node (beta) at (1.5,0.25) {$\beta I(t)$};
        \node (gamma) at (4.5,0.25) {$\gamma$};
    \end{tikzpicture}
    \caption{Visualisation of the SIR model (where $\beta$ is the infection rate, $\gamma$ is the recovery rate, and $I(t)$ is the current number of infective individuals).}
    \label{fig_sir_visual}
\end{figure}
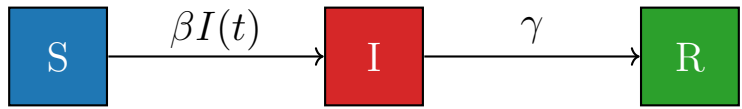

The process is observed by a series of diagnostic tests at prespecified times (each individual $j$ is tested at each time-point $t$ with probability 0.1). Let $y_{t,j} \in \{0,1,\text{NA}\}$ be the test result of individual $j \in \{1,\dots,\Nc\}$ at time-point $t \in \{1,\dots,\Tc\}$, where $y_{t,j}=1$ if they test infective, $y_{t,j}=0$ if they test not infective, and $y_{t,j}=\text{NA}$ if there is no test. Let the test sensitivity and specificity be $s_e$ and $s_p$ respectively; we simulate the result $y_{t,j}$ by sampling \[
    y_{t,j} | x_{t,j} \sim \begin{cases} \text{Bernoulli}(s_e) & \text{if } x_{t,j}=2, \\ \text{Bernoulli}(1-s_p) & \text{if } x_{t,j} \neq 2. \end{cases}
\]

We consider a simulated dataset of $\Nc =100$ individuals over $\Tc=50$ time-points. We set $\beta=1/80$, $\gamma=1/10$, $s_e=0.9$, and $s_p=0.9$. One individual is initially infective, with all others initially susceptible. The simulated hidden transmission process and observation data are shown in Figure \ref{fig_sir_sim}.

We run the MCMC algorithm for $\Kc=10{,}000$ iterations for each method, updating the latent variables 10 times per iteration. The model parameters were kept constant at their true values (in order to better compare the latent variable update methods). Posterior samples of the number of individuals in each state over time using each method (both medians and 95\% credible intervals) are shown in Figure \ref{fig_sir_credible_intervals}. We see that the posterior samples' 95\% credible intervals follow the true values closely; all four methods produce almost identical results.

\begin{figure}[H]
     \centering
     \begin{subfigure}[h]{0.49\textwidth}
         \centering
         \includegraphics[width=\textwidth]{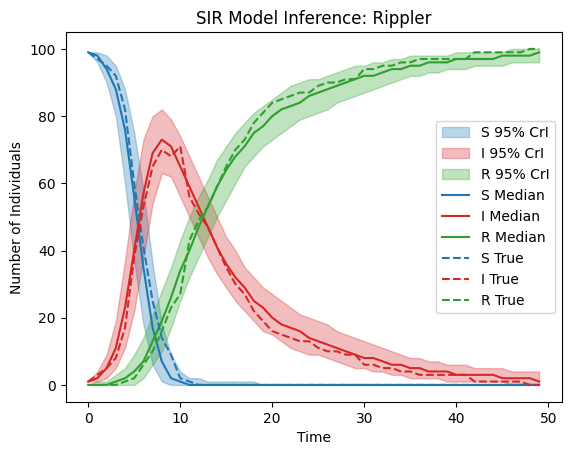}
     \end{subfigure}
     \hfill
     \begin{subfigure}[h]{0.49\textwidth}
         \centering
         \includegraphics[width=\textwidth]{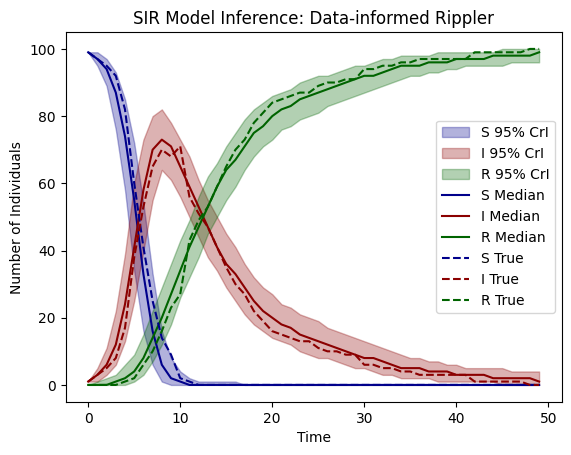}
     \end{subfigure}
     \\ 
     \begin{subfigure}[h]{0.49\textwidth}
         \centering
         \includegraphics[width=\textwidth]{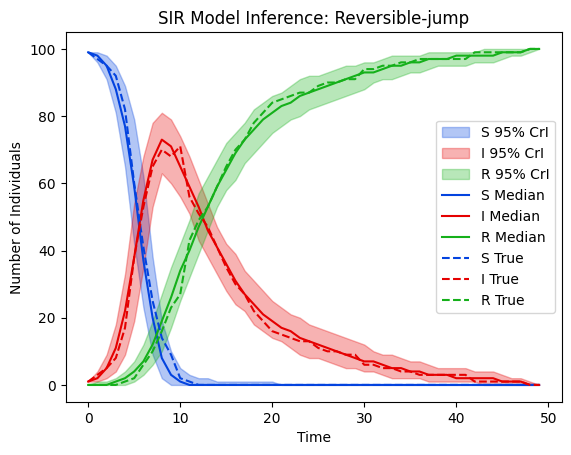}
     \end{subfigure}
     \hfill
     \begin{subfigure}[h]{0.49\textwidth}
         \centering
         \includegraphics[width=\textwidth]{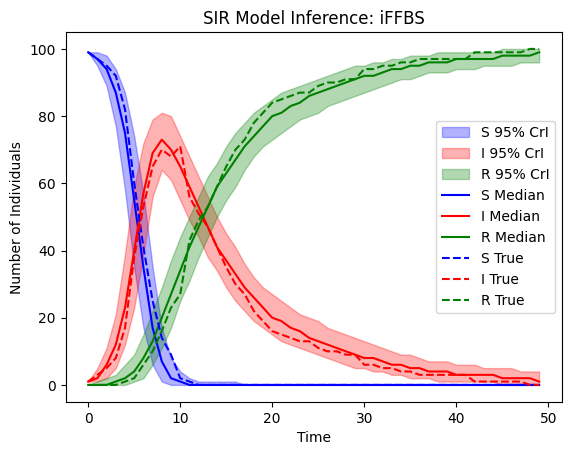}
     \end{subfigure}
     \caption{Posterior samples of the number of individuals in each state over time, using the Rippler (top left), data-informed Rippler (top right), reversible-jump (bottom left), and iFFBS (bottom right) methods on the SIR model. The median values are shown in the solid lines and the 95\% credible intervals are shown in the shaded regions. The true values are shown by the dashed lines.}
     \label{fig_sir_credible_intervals}
\end{figure}

Given that posterior samples generated by the Rippler and data-informed Rippler algorithms converge to the simulated hidden states, we can now focus on comparison between methods of the computational scaling and mixing properties.

\subsection{Extended Exposure Model} \label{sec_sim_sub_SEIR}

The next model we investigate is an SEIR model: an extension of the SIR model with an additional exposed step (for individuals who have contracted the pathogen but are not yet transmitting it to others). We model the length of time an individual is exposed as non-exponential, and so split the exposed step into an arbitrary $\Sc_E$ compartments.

Let $x_{t,j} \in \{1,\dots,\Sc_E+3\}$ be the current state of individual $j \in \{1,\dots,\Nc\}$ at time-point $t \in \{1,\dots,\Tc\}$, where $x_{t,j}=1$ if they are susceptible, $x_{t,j}=i+1$ if they are exposed (step $i$), $x_{t,j}=\Sc_E+2$ if they are infective, and $x_{t,j}=\Sc_E+3$ if they are recovered. State transitions occur from susceptible to infective at rate $\beta$ for each currently infective individual, from each exposed step $i$ to the next exposed step $i+1$ at constant rate $\sigma_i$, and from infective to recovered at constant rate $\gamma$. A visualisation of the model is shown in Figure \ref{fig_extended_exposure_visual}.

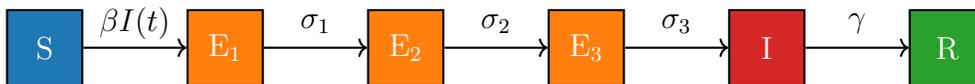
\begin{figure}[H]
    \centering
    \begin{tikzpicture}[thick, scale=1.6, every node/.style={scale=1}]
        \node[thick, draw, rectangle, minimum size=1cm, fill=tab_blue, text=white] (S) at (0,0) {S};
        \node[thick, draw, rectangle, minimum size=1cm, fill=tab_orange, text=white] (E1) at (1.5,0) {E\textsubscript{1}};
        \node[thick, draw, rectangle, minimum size=1cm, fill=tab_orange, text=white] (E2) at (3,0) {E\textsubscript{2}};
        \node[thick, draw, rectangle, minimum size=1cm, fill=tab_orange, text=white] (E3) at (4.5,0) {E\textsubscript{3}};
        \node[thick, draw, rectangle, minimum size=1cm, fill=tab_red, text=white] (I) at (6,0) {I};
        \node[thick, draw, rectangle, minimum size=1cm, fill=tab_green, text=white] (R) at (7.5,0) {R};
        \draw[->,thick] (S)--(E1);
        \draw[->,thick] (E1)--(E2);
        \draw[->,thick] (E2)--(E3);
        \draw[->,thick] (E3)--(I);
        \draw[->,thick] (I)--(R);
        \node (beta) at (0.75,0.2) {$\beta I(t)$};
        \node (sigma1) at (2.25,0.2) {$\sigma_1$};
        \node (sigma2) at (3.75,0.2) {$\sigma_2$};
        \node (sigma3) at (5.25,0.2) {$\sigma_3$};
        \node (gamma) at (6.75,0.2) {$\gamma$};
    \end{tikzpicture}
    \caption{Visualisation of the SEIR model with $\Sc_E=3$ exposure steps.}
    \label{fig_extended_exposure_visual}
\end{figure}

How hidden state inference methods scale with the number of compartments in the model is tested by varying the number of exposed compartments $\Sc_E$, with $\Sc = \Sc_E + 3$ total compartments. We use mean absolute jump distance (MAJD) to measure the mixing of the hidden state inference. The distance between two states is defined as the number of transitions needed to move between them; we have \[\text{MAJD}(\Xb) = \frac{1}{\Kc}\sum_{k=1}^\Kc\sum_{j=1}^\Nc\sum_{t=1}^\Tc \left|x^{(k)}_{t,j}-x^{(k-1)}_{t,j}\right|.\]

As in Section \ref{sec_sim_sub_SIR}, the process is observed by diagnostic tests for infectivity (each individual $j$ and time-point $t$ pair tested with probability 0.1). For each $\Sc_E \in \{1,\dots,7\}$ we simulate a dataset with $\Nc =100$ individuals over $\Tc=100$ time-points. The parameters are set to $\beta=1/50$, $\gamma=1/20$, $s_e=0.8$, $s_p=0.95$, and $\sigma_i = \Sc_E/10$ for each $i \in \{1,\dots,\Sc_E\}$. One individual is initially infective, with all others initially susceptible. The simulated hidden transmission process and observation data when $\Sc_E=3$ are shown in Figure \ref{fig_seir_sim}. Each hidden state inference algorithm (standard Rippler, data-informed Rippler, iFFBS) is run for $\Kc=1000$ iterations, with 10 latent variable updates per iteration. We do not implement the reversible-jump method (because of the high complexity of implementation on a model with an arbitrary number of states). 

Figure \ref{fig_extended_exposure_plots} shows the time taken, and MAJD per unit time, as the number of states $\Sc$ increases. To minimise the impact of implementation, the time taken for each method is relative to the time taken when $\Sc=4$. We see that the iFFBS method is the most efficient when we have few compartments, but that the Rippler methods become more efficient as $\Sc$ increases. 

\begin{figure}[H]
     \centering
     \begin{subfigure}[h]{0.49\textwidth}
         \centering
         \includegraphics[width=\textwidth]{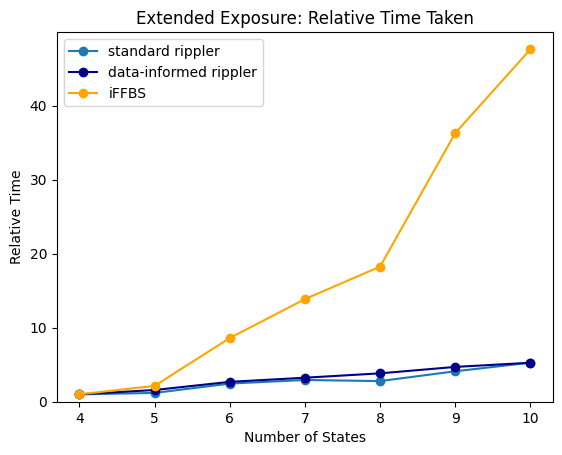}
     \end{subfigure}
     \begin{subfigure}[h]{0.49\textwidth}
         \centering
         \includegraphics[width=\textwidth]{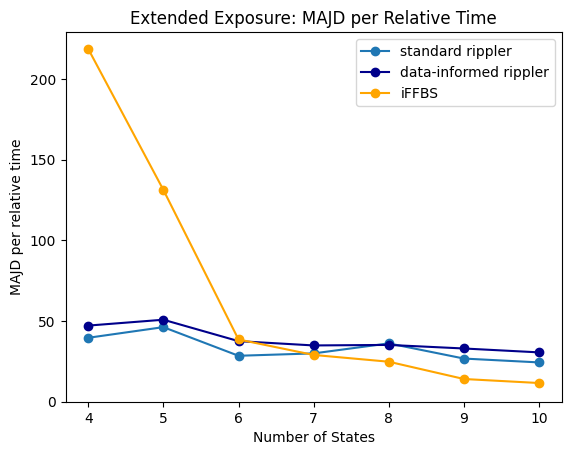}
     \end{subfigure}
     \caption{Relative time taken (left) and MAJD per relative unit time (right) for hidden state inference methods on the extended exposure model (from $\Sc=4$ to $\Sc=10$).}
     \label{fig_extended_exposure_plots}
\end{figure}

\subsection{Multi-strain Model} \label{sec_sim_sub_SIS}

In this section we investigate an additional model that has an arbitrary number of states: a multi-strain model based on \cite{touloupou2020bayesian}. This is an extension of the susceptible-infective-susceptible (SIS) model, splitting the infective compartment into $\Sc_I$ mutually-exclusive strains.

Let $x_{t,j} \in \{1,\dots,\Sc_I+1\}$ be the current state of individual $j \in \{1,\dots,\Nc\}$ at time-point $t \in \{1,\dots,\Tc\}$, where $x_{t,j}=1$ if they are susceptible and $x_{t,j}=i+1$ if they are infected (with strain $i$). State transitions occur from susceptible to infected (with strain $i$) at time-varying rate $\lambda_i(t)$ and from infected (with strain $i$) to recovered at constant rate $\gamma_i$. Individuals can also transition between infected states; the rate of transitioning into strain $i$ is reduced to $\delta\lambda_i(t)$, where $\delta \in [0,1]$ is some constant. For each $i \in \{1,\dots,\Sc_I\}$ we let $\lambda_i(t) = \beta_i I_i(t)$, where $\beta_i$ is a constant and $I_i(t) = \sum_{i=j}^\Nc \mathbf{1}\{x_{t,j}=i+1\}$ (the total number of individuals currently infected with strain $i$). A visualisation of the model is shown in Figure \ref{fig_multi_strain_visual}.

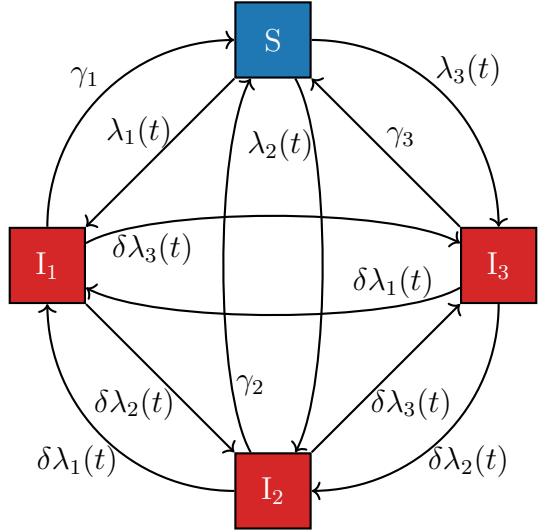
\begin{figure}[H]
    \centering
    \begin{tikzpicture}[thick, scale=2, every node/.style={scale=1}]
        \node[thick, draw, rectangle, minimum size=1cm, fill=tab_blue, text=white] (S) at (0,1.5) {S};
        \node[thick, draw, rectangle, minimum size=1cm, fill=tab_red, text=white] (I1) at (-1.5,0) {I\textsubscript{1}};
        \node[thick, draw, rectangle, minimum size=1cm, fill=tab_red, text=white] (I2) at (0,-1.5) {I\textsubscript{2}};
        \node[thick, draw, rectangle, minimum size=1cm, fill=tab_red, text=white] (I3) at (1.5,0) {I\textsubscript{3}};
        \path[->] (S) edge (I1);
        \node () at (-0.89,0.89) {$\lambda_1(t)$};
        \path[->] (I1) edge [bend left, looseness=1, out=45, in=135] (S);
        \node () at (-1.25,1.25) {$\gamma_1$};
        \path[->] (S) edge [bend left, looseness=0.5] (I2);
        \node () at (0.05,0.8) {$\lambda_2(t)$};
        \path[->] (I2) edge [bend left, looseness=0.5] (S);
        \node () at (-0.15,-0.8) {$\gamma_2$};
        \path[->] (S) edge [bend left, looseness=1, out=45, in=135] (I3);
        \node () at (1.3,1.3) {$\lambda_3(t)$};
        \path[->] (I3) edge (S);
        \node () at (0.85,0.85) {$\gamma_3$};
        \path[->] (I1) edge (I2);
        \node () at (-0.92,-0.92) {$\delta\lambda_2(t)$};
        \path[->] (I2) edge [bend left, looseness=1, out=45, in=135] (I1);
        \node () at (-1.3,-1.3) {$\delta\lambda_1(t)$};
        \path[->] (I1) edge [bend left, looseness=0.5] (I3);
        \node () at (-0.8,0.11) {$\delta\lambda_3(t)$};
        \path[->] (I3) edge [bend left, looseness=0.5] (I1);
        \node () at (0.8,-0.11) {$\delta\lambda_1(t)$};
        \path[->] (I2) edge (I3);
        \node () at (0.92,-0.92) {$\delta\lambda_3(t)$};
        \path[->] (I3) edge [bend left, looseness=1, out=45, in=135] (I2);
        \node () at (1.3,-1.3) {$\delta\lambda_2(t)$};
    \end{tikzpicture}
    \caption{Visualisation of the multi-strain SIS model with $\Sc_I=3$ strains.}
    \label{fig_multi_strain_visual}
\end{figure}


As in Section \ref{sec_sim_sub_SIR}, the process is observed by diagnostic tests for infectivity (each individual $j$ and time-point $t$ pair tested with probability 0.1). For each $\Sc_I \in \{3,\dots,9\}$ we simulate a dataset with $\Nc =40$ individuals over $\Tc=50$ time-points. The parameters are set to $\delta=0.2$, $s_e=0.8$, $s_p=0.95$, and $\beta_i = 1/100$, $\gamma_i = 1/10$ for each $i \in \{1,\dots,\Sc_I\}$. One individual is initially infective with each strain, with all others initially susceptible. The simulated hidden transmission process and observation data when $\Sc_I=3$ are shown in Figure \ref{fig_sis_sim}. Each hidden state inference algorithm (standard Rippler, data-informed Rippler, iFFBS) is run for $\Kc=1000$ iterations, with 10 latent variable updates per iteration.

Figure \ref{fig_multi_strain_plots} shows the time taken, and mean absolute jump distance per unit time, as the number of states $\Sc$ increases. The mean absolute jump distance is calculated as \[\text{MAJD}(\Xb) = \frac{1}{\Kc}\sum_{k=1}^\Kc\sum_{j=1}^\Nc\sum_{t=1}^\Tc \mathbf{1}\left\{x^{(k)}_{t,j} \neq x^{(k-1)}_{t,j}\right\}.\] To minimise the impact of implementation, the time taken for each method is relative to the time taken when $\Sc=4$. We see that Rippler and data-informed Rippler methods scale better with the number of compartments than the iFFBS method.

\begin{figure}[H]
     \centering
     \begin{subfigure}[h]{0.49\textwidth}
         \centering
         \includegraphics[width=\textwidth]{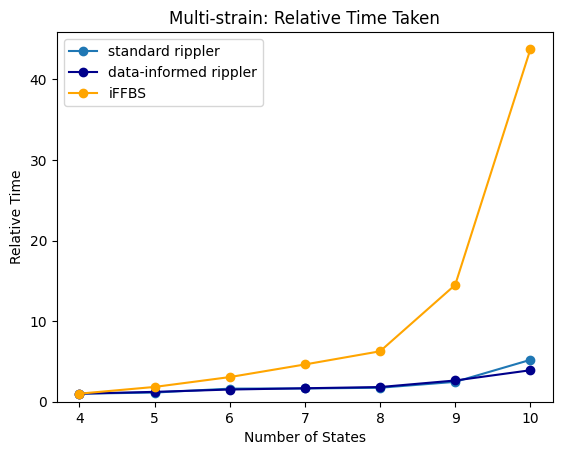}
     \end{subfigure}
     \begin{subfigure}[h]{0.49\textwidth}
         \centering
         \includegraphics[width=\textwidth]{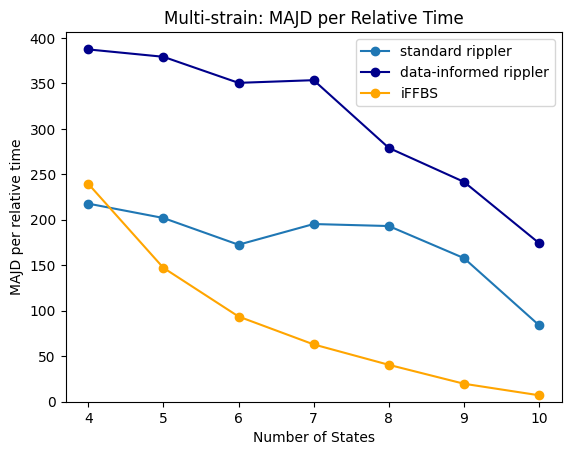}
     \end{subfigure}
     \caption{Relative time taken (left) and MAJD per relative unit time (right) for hidden state inference methods on the multi-strain
     model (from $\Sc=4$ to $\Sc=10$).}
     \label{fig_multi_strain_plots}
\end{figure}

\section{Discussion} \label{sec_dis}

In this article, we have developed two novel MCMC algorithms for Bayesian inference on the hidden states of a CHMM: the Rippler algorithm (an extension from \cite{neill2025non}) and the data-informed Rippler algorithm. These algorithms use the hidden state structure in the proposal of new hidden states (with the data-informed Rippler algorithm also using the observation data in the proposal step).

Both Rippler algorithms can be used for inference on any individual-based state-transition epidemic model or CHMM. Which method is more efficient depends on the transmission model and observation data -- we found that the two methods had similar efficiency for the extended exposure model (Figure \ref{fig_extended_exposure_plots}), but that the data-informed method was more efficient for the multi-strain model (Figure \ref{fig_multi_strain_plots}). This is because in the multi-strain model transitions are possible between any two states, so the data-informed Rippler method is more able to match the hidden state to the data. Note that the standard Rippler method was developed for diagnostic test data, but not for exact observations (e.g., transition times into a given state). However, the data-informed Rippler method is viable in this case and performs well (see Section \ref{sec_supp_sim_sub_partially_known}).

We have also compared the Rippler algorithms to the state-of-the-art methods for inference on partially observed stochastic epidemics -- RJMCMC \citep{gibson1998estimating, o1999bayesian} and iFFBS \citep{touloupou2020scalable} -- and found that they perform better than these other methods as the number of disease states in the model increases. In particular, the computational complexity of the Rippler algorithms scale linearly with the number of disease states, while computational complexity iFFBS scales quadratically with the number of disease states.

The Rippler algorithms can currently only be applied to individual-based stochastic epidemic models. However, population- and metapopulation-based epidemic models are also of interest. Simulating transmission from these models requires sampling from a multinomial distribution rather than a categorical distribution; we anticipate that any difficulty in extending the Rippler methods to these models will be in non-centering the multinomial distribution. Note that in the case where there is no more than one transition out any disease state, we only need to sample from a binomial distribution (simplifying the non-centering). 

\section*{Acknowledgments}

This paper is based on work completed while James Neill was part of the EPSRC funded STOR-i centre for doctoral training EP/S022252/1.

\section*{Supplementary Material}

The Supplementary Material includes further details on the Rippler algorithm and additional simulation studies. The Python code used to run the simulation and inference can be found at \url{https://github.com/neilljn/Generalised_Rippler}.

\bibliography{bibliography.bib}

\newpage

{\LARGE Non-centred Bayesian inference for discrete-valued epidemic models: the Rippler algorithm -- Supplementary Materials}

\renewcommand{\thefigure}{S\arabic{figure}}
\renewcommand{\theequation}{S\arabic{equation}}
\renewcommand{\thetable}{S\arabic{table}}
\renewcommand{\thealgorithm}{S\arabic{algorithm}}
\renewcommand{\thesection}{S\arabic{section}}
\renewcommand{\theproposition}{S\arabic{proposition}}
\setcounter{equation}{0} \setcounter{figure}{0} \setcounter{table}{0} \setcounter{section}{0}

\section{Non-centered Simulation} \label{sec_supp_noncentred}

In Section \ref{sec_model_sub_sim} we discuss taking a non-centred approach to drawing from a Categorical distribution when determining $\Xb$. We define $g$ as the function comparing a random number $u$ to the transition probabilities $(p_1,\dots,p_\Sc)$ in this process, returning $s \in \{1,\dots,\Sc\}$. 

The cumulative distribution function for a Categorical$(\Sc,(p_1,\dots,p_\Sc))$ distribution is \[ F(x) = \begin{cases}
    0 & \text{ if } x < 1,
    \\ \sum_{r=1}^{\lfloor x \rfloor} p_r & \text{ if } 1 \leq x < \Sc,
    \\ 1 & \text{ if } x \geq \Sc.
\end{cases} \]

This means $g(u,\Sc,(p_1,\dots,p_\Sc))$ is equal to the value of $s \in \{1,\dots,\Sc\}$ such that \begin{equation}
    \sum_{r=1}^{s-1} p_r < u < \sum_{r=1}^{s} p_r, \label{u_sim_bounds}
\end{equation} or equivalently, \[ g(u,\Sc,(p_1,\dots,p_\Sc)) = 1 + \sum_{s=1}^\Sc \textbf{1}\left\{u>\sum_{r=1}^sp_r\right\}. \]

\section{Bayesian Inference} \label{sec_supp_rippler}

\subsection{The Rippler Algorithm} \label{sec_supp_rippler_sub_rippler}

In Section \ref{sec_rippler_sub_proposal} we discuss calculating matrices of lower and upper bounds on the random numbers that could produce the latent variables. These values are determined based on the inverse of the non-centred simulation (equation \ref{u_sim_bounds}). For the initial hidden states, we calculate \begin{equation}
    u^{low}_{1,j} = \sum_{r=1}^{x_{1,j}-1}\tilde{p}^{(j)}_r, \hspace{5mm} u^{upp}_{1,j} = \sum_{r=1}^{x_{1,j}}\tilde{p}^{(j)}_r \label{u_bounds_initial}
\end{equation} for each individual $j \in \{1,\dots,\Nc\}$. Then for $t \in \{1,\dots,\Tc-1\}$ we calculate \begin{equation}
    u^{low}_{t+1,j} = \sum_{r=1}^{x_{t+1,j}-1}p^{(t,j)}_{x_{t,j},r}, \hspace{5mm} u^{upp}_{t+1,j} = \sum_{r=1}^{x_{t+1,j}}p^{(t,j)}_{x_{t,j},r} \label{u_bounds_full}
\end{equation} for every individual $j \in \{1,\dots,\Nc\}$. 


The proof of Proposition \ref{prop1} (from Section \ref{sec_rippler_sub_acceptance}) is given below.

\textit{Proof of Proposition \ref{prop1}.} We wish to show $\pi(\Xb|\thetab)q_1(\Ub|\thetab,\Xb) = 1$ and $\pi(\Xb^{\ast}|\thetab)q_1(\Ub^{\ast}|\thetab,\Xb^{\ast}) = 1$. From our transmission model we have \begin{equation}
    \pi(\Xb|\thetab) = \prod_{j=1}^\Nc \left( \tilde{p}^{(j)}_{x_{1,j}} \prod_{t=2}^\Tc p^{(t-1,j)}_{x_{t-1,j},x_{t,j}} \right). \label{X|theta}
\end{equation}

Each element $u_{t,j}$ is proposed uniformly from the range $(u^{low}_{t,j},u^{upp}_{t,j})$. From equations \ref{u_bounds_initial} and \ref{u_bounds_full} the probability density of the proposal at $(t,j)$ is \begin{align*}
    \frac{1}{u^{upp}_{t,j}-u^{low}_{t,j}} &= \begin{dcases}
        \frac{1}{\sum_{r=1}^{x_{1,j}}\tilde{p}^{(j)}_r - \sum_{r=1}^{x_{1,j}-1}\tilde{p}^{(j)}_r} & \text{if }t = 1,
        \\ \frac{1}{\sum_{r=1}^{x_{t,j}}p^{(t-1,j)}_{x_{t-1,j},r}-\sum_{r=1}^{x_{t,j}-1}p^{(t-1,j)}_{x_{t-1,j},r}} & \text{if } t \in\{2,\dots,\Tc\},
    \end{dcases}
    \\ &= \begin{dcases}
        \frac{1}{\tilde{p}^{(j)}_{x_{1,j}}} & \text{if }t = 1,
        \\ \frac{1}{p^{(t-1,j)}_{x_{t-1,j},x_{t,j}}} & \text{if }t \in \{2,\dots,\Tc\}.
    \end{dcases}
\end{align*}

This means the overall proposal density is \begin{equation}
    q_1(\Ub|\thetab,\Xb) = \prod_{j=1}^\Nc \left( \frac{1}{\tilde{p}^{(j)}_{x_{1,j}}} \prod_{t=2}^\Tc \frac{1}{p^{(t-1,j)}_{x_{t-1,j},x_{t,j}}} \right). \label{U|theta,X}
\end{equation}

Multiplying equation \ref{X|theta} by \ref{U|theta,X} the result follows.  By symmetry, the same result follows for the reverse move. \begin{flushright} $\Box$ \end{flushright}

\subsection{Data-informed Rippler} \label{sec_supp_rippler_sub_hybrid}

The modified probabilities for the data-informed Rippler algorithm are calculated as \begin{align*}
    \tilde{\rho}^{(j)}_{s} &= \Pbb(x_{1,j}=s|y_{1,j},\thetab)
    \\ &= \frac{\Pbb(x_{1,j}=s|\thetab) f(y_{1,j}|x_{1,j}=s,\thetab)}{f(y_{1,y}|\thetab)} \hspace{2cm} &\text{Bayes' theorem}
    \\ &= \frac{\Pbb(x_{1,j}=s|\thetab) f(y_{1,j}|x_{1,j}=s,\thetab)}{\sum_{s'=1}^\Sc \Pbb(x_{1,j}=s'|\thetab) f(y_{1,j}|x_{1,j}=s',\thetab)} &\text{law of total probability}
    \\ &= \frac{\tilde{p}^{(j)}_{s}f(y_{1,j}|x_{1,j}=s,\thetab)}{\sum_{s'=1}^\Sc \tilde{p}^{(j)}_{s'} f(y_{1,j}|x_{1,j}=s',\thetab)}
\end{align*} and \begin{align*}
    \rho^{(t,j)}_{r,s} &= \Pbb(x_{t+1,j}=s|x_{t,j}=r,\boldsymbol{x}_{t,-j},y_{t+1,j},\thetab)
    \\ &= \frac{\Pbb(x_{t+1,j}=s|x_{t,j}=r,\boldsymbol{x}_{t,-j},\thetab) f(y_{t+1,j}|x_{t+1,j}=s,x_{t,j}=r,\boldsymbol{x}_{t,-j},\thetab)}{f(y_{t+1,j}|x_{t,j}=r,\boldsymbol{x}_{t,-j},\thetab)} 
    \\ &\hspace{13.2cm} \text{Bayes' theorem}
    \\ &= \frac{\Pbb(x_{t+1,j}=s|x_{t,j}=r,\boldsymbol{x}_{t,-j},\thetab) f(y_{t+1,j}|x_{t+1,j}=s,x_{t,j}=r,\boldsymbol{x}_{t,-j},\thetab)}{\sum_{s'=1}^\Sc \Pbb(x_{t+1,j}=s'|x_{t,j}=r,\boldsymbol{x}_{t,-j},\thetab) f(y_{t+1,j}|x_{t+1,j}=s',x_{t,j}=r,\boldsymbol{x}_{t,-j},\thetab)} 
    \\ &\hspace{11.75cm} \text{law of total probability}
    \\ &= \frac{p^{(t,j)}_{r,s}f(y_{t+1,j}|x_{t+1,j}=s,\thetab)}{\sum_{s'=1}^\Sc p^{(t,j)}_{r,s'} f(y_{t+1,j}|x_{t+1,j}=s',\thetab)}. 
    \\ &\hspace{8.1cm} y_{t+1,j}|x_{t+1,j},\thetab \text{ independent of }x_{t,j}\text{ and }\boldsymbol{x}_{t,-j}
\end{align*}

We now present the modifications to equations \ref{u_bounds_initial} and \ref{u_bounds_full} when using the data-informed Rippler algorithm from Section \ref{sec_rippler_sub_hybrid}. For the initial hidden states, we calculate \begin{equation}
    u^{low}_{1,j} = \sum_{r=1}^{x_{1,j}-1}\tilde{\rho}^{(j)}_r, \hspace{5mm} u^{upp}_{1,j} = \sum_{r=1}^{x_{1,j}}\tilde{\rho}^{(j)}_r \label{u_bounds_initial_hybrid}
\end{equation} for each individual $j \in \{1,\dots,\Nc\}$. Then for $t \in \{1,\dots,\Tc-1\}$ we calculate \begin{equation}
    u^{low}_{t+1,j} = \sum_{r=1}^{x_{t+1,j}-1}\rho^{(t,j)}_{x_{t,j},r}, \hspace{5mm} u^{upp}_{t+1,j} = \sum_{r=1}^{x_{t+1,j}}\rho^{(t,j)}_{x_{t,j},r} \label{u_bounds_full_hybrid}
\end{equation} for every individual $j \in \{1,\dots,\Nc\}$.

The proof of Proposition \ref{prop3} (from Section \ref{sec_rippler_sub_hybrid}) is given below.

\textit{Proof of Proposition \ref{prop3}.} From our transmission model we have \begin{equation}
    \pi(\Yb|\thetab,\Xb) = \prod_{t=1}^\Tc \prod_{j=1}^\Nc f(y_{t,j}|x_{t,j},\thetab). \label{Y|X,theta}
\end{equation}

Each element $u_{t,j}$ is proposed uniformly from the range $(u^{low}_{t,j},u^{upp}_{t,j})$. From equations \ref{u_bounds_initial_hybrid} and \ref{u_bounds_full_hybrid} the probability density of the proposal at $(t,j)$ is \begin{align*}
    \frac{1}{u^{upp}_{t,j}-u^{low}_{t,j}} &= \begin{dcases}
        \frac{1}{\sum_{r=1}^{x_{1,j}}\tilde{\rho}^{(j)}_r - \sum_{r=1}^{x_{1,j}-1}\tilde{\rho}^{(j)}_r} & \text{if }t = 1,
        \\ \frac{1}{\sum_{r=1}^{x_{t,j}}\rho^{(t-1,j)}_{x_{t-1,j},r}-\sum_{r=1}^{x_{t,j}-1}\rho^{(t-1,j)}_{x_{t-1,j},r}} & \text{if } t \in\{2,\dots,\Tc\},
    \end{dcases}
    \\ &= \begin{dcases}
        \frac{1}{\tilde{\rho}^{(j)}_{x_{1,j}}} & \text{if }t = 1,
        \\ \frac{1}{\rho^{(t-1,j)}_{x_{t-1,j},x_{t,j}}} & \text{if }t \in \{2,\dots,\Tc\},
    \end{dcases}
    \\ &= \begin{dcases}
        \frac{\tilde{c}^{(j)}}{\tilde{p}^{(j)}_{x_{1,j}}f(y_{1,j}|x_{1,j})} & \text{if }t = 1,
        \\ \frac{c^{(t-1,j)}_{x_{t-1,j}}}{p^{(t-1,j)}_{x_{t-1,j},x_{t,j}}f(y_{t,j}|x_{t,j})} & \text{if }t \in \{2,\dots,\Tc\}.
    \end{dcases}
\end{align*}

This means the overall proposal density is \begin{equation}
    q_1(\Ub|\thetab,\Xb) = \prod_{j=1}^\Nc \left( \frac{\tilde{c}^{(j)}}{\tilde{p}^{(j)}_{x_{1,j}}f(y_{1,j}|x_{1,j})} \prod_{t=2}^\Tc \frac{c^{(t-1,j)}_{x_{t-1,j}}}{p^{(t-1,j)}_{x_{t-1,j},x_{t,j}}f(y_{t,j}|x_{t,j})} \right). \label{U|theta,X hybrid}
\end{equation}

Multiplying equations \ref{X|theta} by \ref{Y|X,theta} by \ref{U|theta,X hybrid} the result follows.  By symmetry, the same result follows for the reverse move. \begin{flushright} $\Box$ \end{flushright}

\section{Simulation Studies} \label{sec_supp_sim}

\subsection{SIR Model} \label{sec_supp_sim_sub_SIR}

In this section we provide additional results from the simulation study in Section \ref{sec_sim_sub_SIR}. The simulation of the SIR model used as the test dataset, including both the hidden transmission process and observed test results, is shown in Figure \ref{fig_sir_sim}. 

\begin{figure}[H]
     \centering
     \begin{subfigure}[h]{0.49\textwidth}
         \centering
         \includegraphics[width=\textwidth]{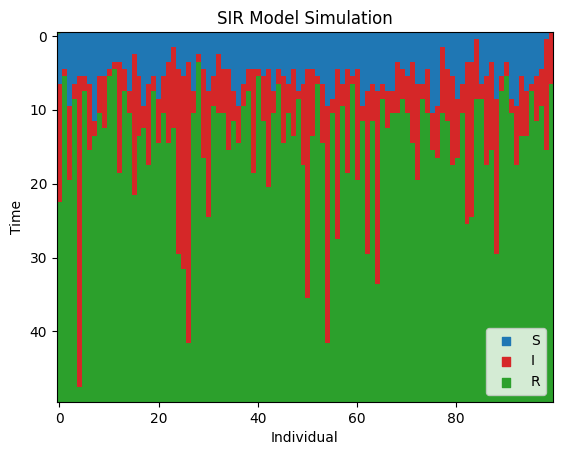}
     \end{subfigure}
     \hfill
     \begin{subfigure}[h]{0.49\textwidth}
         \centering
         \includegraphics[width=\textwidth]{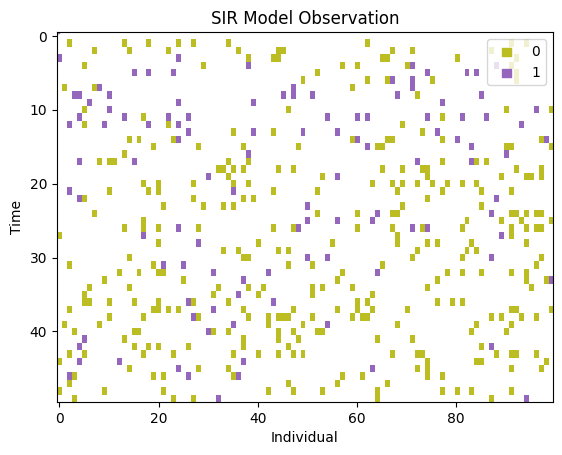}
     \end{subfigure}
     \caption{Simulation of the SIR model: hidden transmission process (left) and observation data (right).}
     \label{fig_sir_sim}
\end{figure}

The frequency of different ripple sizes in the hidden state inference is shown in Figure \ref{fig_sir_ripple_sizes} (ripple sizes between 1 and 100). The size of the ripple is defined as the number of hidden states $x_{t,j}$ changed between the current and proposed latent variables. Figure \ref{fig_sir_ripple_sizes} shows that the data-informed Rippler algorithm both proposes and accepts larger moves than the standard Rippler. Adaptive tuning is used to determine how many elements of $\Ub$ to change in each update (aiming to optimise the proposed ripple sizes). We set exploration probability $\epsilon=0.05$, maximum change $\kappa_{max}=10$, and target acceptance rate $a'=0.234$. The standard Rippler algorithm chooses $\kappa=3$ in $31.4\%$ of updates, $\kappa=4$ in $61.0\%$ of updates, and other $\kappa$ in $4.9\%$ of updates. The data-informed Rippler algorithm chooses $\kappa=5$ in $77.9\%$ of updates, $\kappa=6$ in $16.6\%$ of updates, and other $\kappa$ in $5.5\%$ of updates.

\begin{figure}[H]
     \centering
     \begin{subfigure}[h]{0.49\textwidth}
         \centering
         \includegraphics[width=\textwidth]{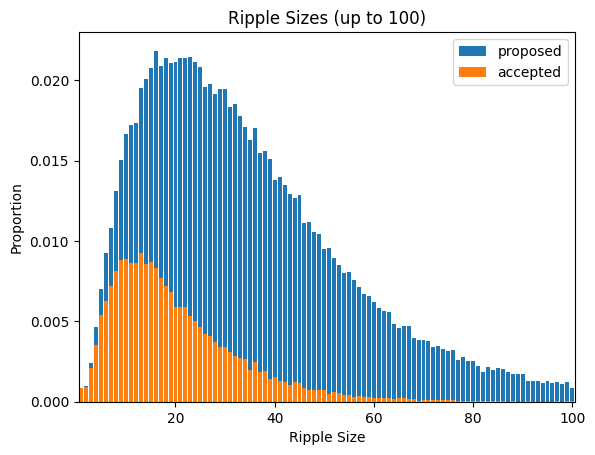}
     \end{subfigure}
     \hfill
     \begin{subfigure}[h]{0.49\textwidth}
         \centering
         \includegraphics[width=\textwidth]{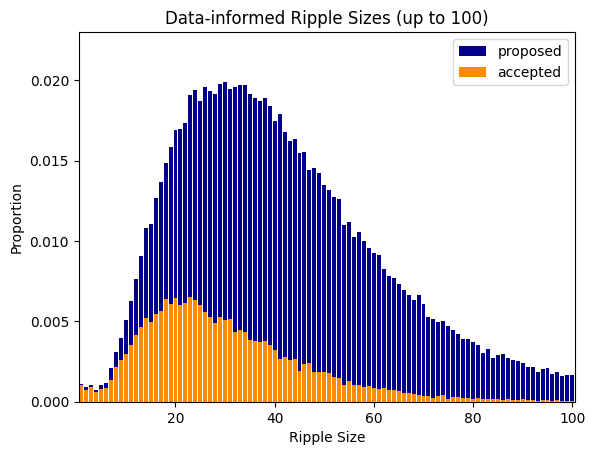}
     \end{subfigure}
     \caption{The frequency of ripple sizes when using the Rippler (left) and data-informed Rippler (right) methods for hidden state inference. The frequencies of proposed and accepted moves are in blue and orange respectively.}
     \label{fig_sir_ripple_sizes}
\end{figure}

Figure \ref{fig_sir_accept_MAJD} compares the acceptance rate and MAJD when varying the number of elements of $\Ub$ changed (when proposing $\Ub^\ast$ in the hidden state inference). Note that in Figure \ref{fig_sir_accept_MAJD} the number of elements changed in each proposal is kept the same for the full MCMC sampling (not using adaptive tuning). We see that the MAJD is maximised at a point adjacent to the target acceptance rate for the adaptive algorithm, 0.234.

\begin{figure}[H]
     \centering
     \begin{subfigure}[h]{0.49\textwidth}
         \centering
         \includegraphics[width=\textwidth]{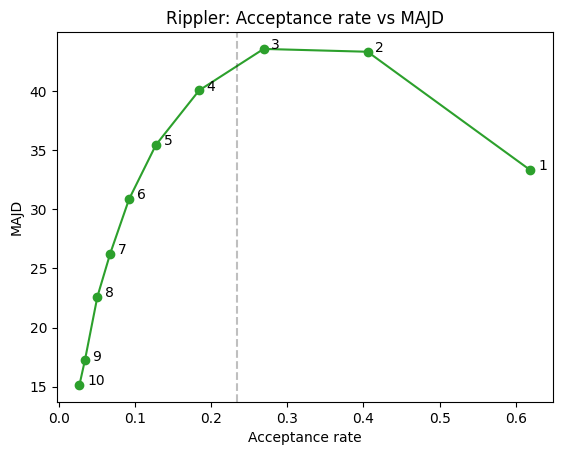}
     \end{subfigure}
     \hfill
     \begin{subfigure}[h]{0.49\textwidth}
         \centering
         \includegraphics[width=\textwidth]{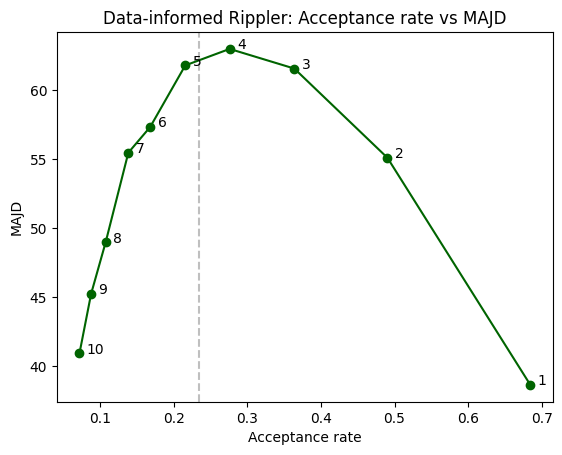}
     \end{subfigure}
     \caption{Comparison of the acceptance rate and MAJD when using the Rippler (left) and data-informed Rippler (right) methods for hidden state inference. The label next to each node is the number of elements of $\Ub$ changed when proposing $\Ub^\ast$ (see step 7 of Algorithms \ref{alg_rippler} and \ref{alg_rippler_data_informed}). The black dashed line is at acceptance rate 0.234, the the adaptive tuning target acceptance rate (see Section \ref{sec_rippler_sub_tuning}).}
     \label{fig_sir_accept_MAJD}
\end{figure}

\subsection{SIR Model: Partially Known Transition Times} \label{sec_supp_sim_sub_partially_known}

In Section \ref{sec_sim_sub_SIR} we assess the performance of the Rippler and data-informed Rippler algorithms on the general stochastic SIR epidemic, using a series of diagnostic tests as observation data. However, the set of observation data can instead be a partial set of known transition event times (with other transition times unknown) -- we now test the data-informed Rippler algorithm on the SIR model in this case (and compare it with iFFBS). Note that we do not use the standard Rippler method here -- it would likely propose hidden states incompatible with the known transition event times, making it very inefficient. The data-informed Rippler method can use the observations in the hidden state proposal, so will be viable with partially known transition times.

The SIR transmission simulation used is the same as in Section Section \ref{sec_sim_sub_SIR}. Recovery event times are used as the observation data (with infection event times unknown). In order to use the data-informed Rippler algorithm, we must express the recovery times as the observation data in a CHMM. Let $\tilde{t}_j$ be the recovery time of each individual $j$. We define \[ y_{t,j} = \begin{cases}
    \text{S/I} & \text{if } t \leq \tilde{t}_j -2  ~~~\text{(either infective or susceptible),}
    \\ \text{I} & \text{if } t = \tilde{t}_j - 1  ~~~\text{(known infective),}
    \\ \text{R} & \text{if } t \geq \tilde{t}_j ~~~\text{(known recovered).}
\end{cases} \] 

We run both the data-informed Rippler and iFFBS algorithms for $\Kc=10{,}000$ iterations, updating the latent variables 10 times per iteration. Posterior samples (both medians and 95\% credible intervals) of the number of individuals in each state over time using both methods are shown in Figure \ref{fig_sir_credible_intervals_known_recovery}. We see that the posterior samples' $95\%$ credible intervals follow the true values closely for both methods. The mixing of the hidden state inference, as measured by MAJD, is significantly greater when using the data-informed Rippler algorithm (108.1) than when using iFFBS (18.5). 

\begin{figure}[H]
     \centering
     \begin{subfigure}[h]{0.49\textwidth}
         \centering
         \includegraphics[width=\textwidth]{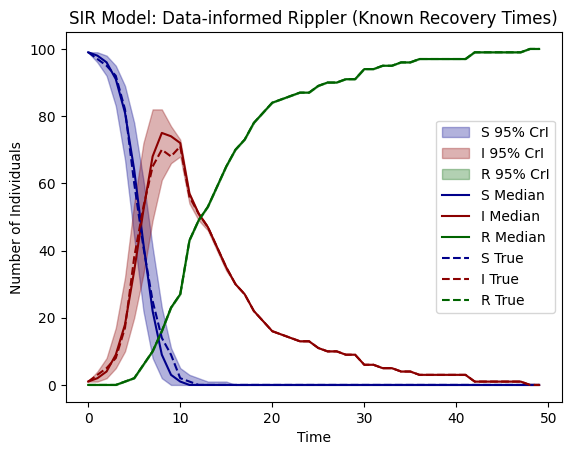}
     \end{subfigure}
     \hfill
     \begin{subfigure}[h]{0.49\textwidth}
         \centering
         \includegraphics[width=\textwidth]{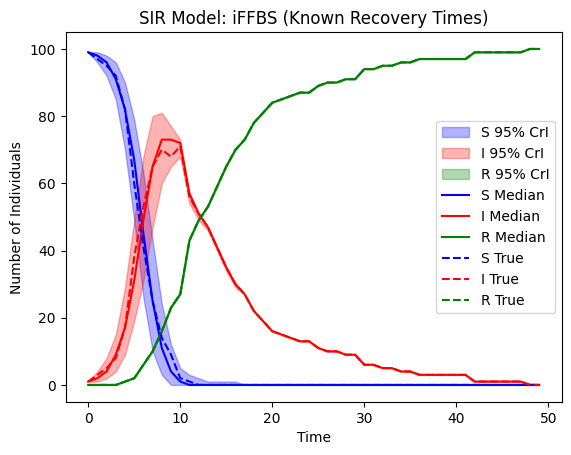}
     \end{subfigure}
     \caption{Posterior samples of the number of individuals in each state over time, using the data-informed Rippler (left), and iFFBS (right) methods on the SIR model (using recovery times as the observed data, instead of diagnostic tests). The median values are shown in the solid lines and the 95\% credible intervals are shown in the shaded regions. The true values are shown by the dashed lines.}
     \label{fig_sir_credible_intervals_known_recovery}
\end{figure}

\subsection{Extended Exposure Model} \label{sec_supp_sim_sub_SEIR}

The simulation of the extended exposure SEIR model used as the test dataset, including both the hidden transmission process and observed test results, is shown in Figure \ref{fig_seir_sim}. Posterior samples (both medians and 95\% credible intervals) of the number of individuals in each state over time using each hidden state inference method are shown in Figure \ref{fig_seir_credible_intervals}.

\begin{figure}[H]
     \centering
     \begin{subfigure}[h]{0.49\textwidth}
         \centering
         \includegraphics[width=\textwidth]{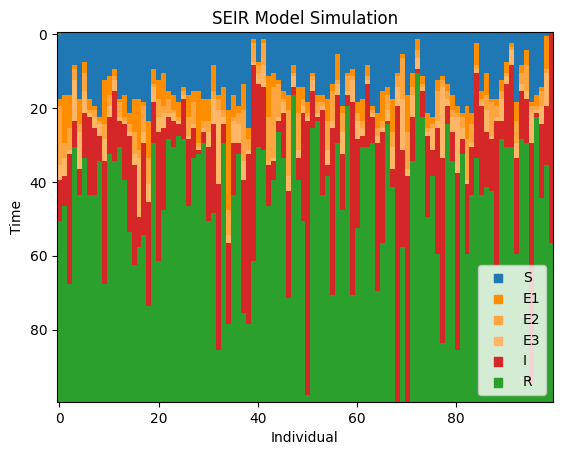}
     \end{subfigure}
     \hfill
     \begin{subfigure}[h]{0.49\textwidth}
         \centering
         \includegraphics[width=\textwidth]{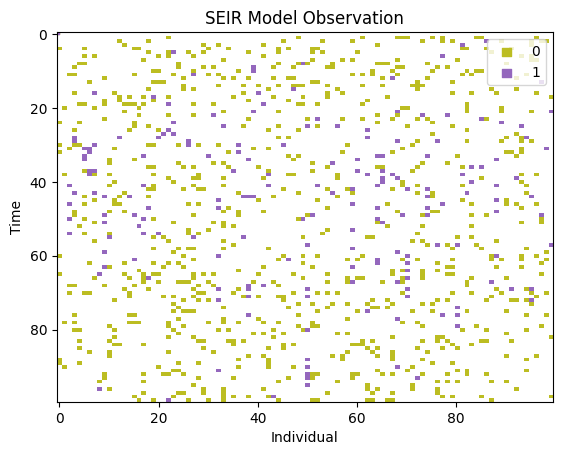}
     \end{subfigure}
     \caption{Simulation of the extended exposure model with $\Sc_E=3$ exposure steps: hidden transmission process (left) and observation data (right).}
     \label{fig_seir_sim}
\end{figure}

\begin{figure}[H]
     \centering
     \begin{subfigure}[h]{0.32\textwidth}
         \centering
         \includegraphics[width=\textwidth]{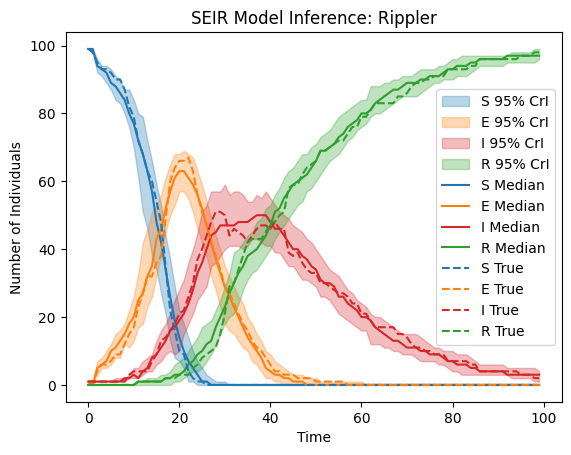}
     \end{subfigure}
     \hfill
     \begin{subfigure}[h]{0.32\textwidth}
         \centering
         \includegraphics[width=\textwidth]{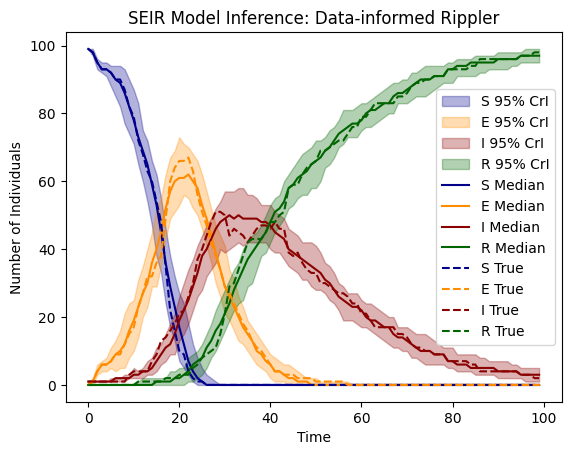}
     \end{subfigure}
     \hfill
     \begin{subfigure}[h]{0.32\textwidth}
         \centering
         \includegraphics[width=\textwidth]{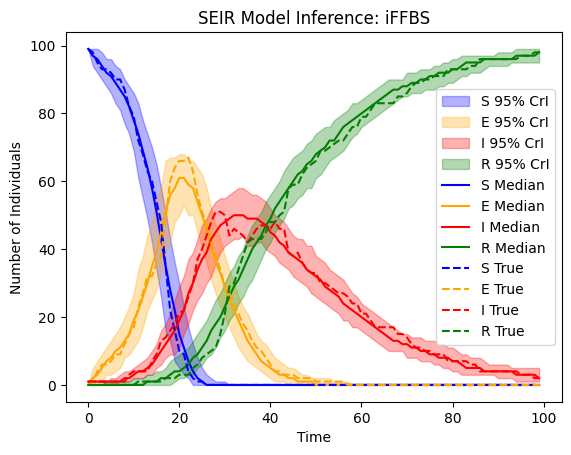}
     \end{subfigure}
     \caption{Posterior samples of the number of individuals in each state over time, using the Rippler (left), data-informed Rippler (middle), and iFFBS (right) methods on the extended exposure model ($\Sc_E=3$ exposure steps). The median values are shown in the solid lines and the 95\% credible intervals are shown in the shaded regions. The true values are shown by the dashed lines.}
     \label{fig_seir_credible_intervals}
\end{figure}

\subsection{Multi-strain Model} \label{sec_supp_sim_sub_SIS}

The simulation of the multi-strain SIS model used as the test dataset, including both the hidden transmission process and observed test results, is shown in Figure \ref{fig_sis_sim}. Posterior samples (both medians and 95\% credible intervals) of the number of individuals in each state over time using each hidden state inference method are shown in Figure \ref{fig_sis_credible_intervals}.

\begin{figure}[H]
     \centering
     \begin{subfigure}[h]{0.49\textwidth}
         \centering
         \includegraphics[width=\textwidth]{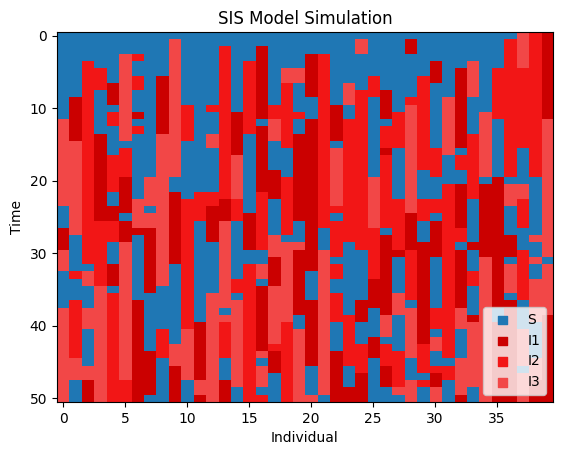}
     \end{subfigure}
     \hfill
     \begin{subfigure}[h]{0.49\textwidth}
         \centering
         \includegraphics[width=\textwidth]{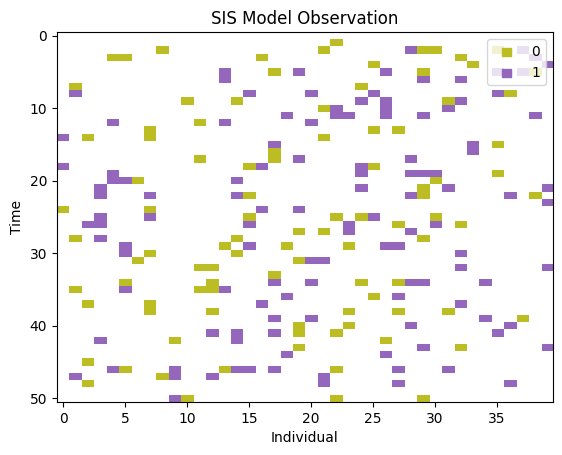}
     \end{subfigure}
     \caption{Simulation of the multi-strain model with $\Sc_I=3$ strains: hidden transmission process (left) and observation data (right).}
     \label{fig_sis_sim}
\end{figure}

\begin{figure}[H]
     \centering
     \begin{subfigure}[h]{0.32\textwidth}
         \centering
         \includegraphics[width=\textwidth]{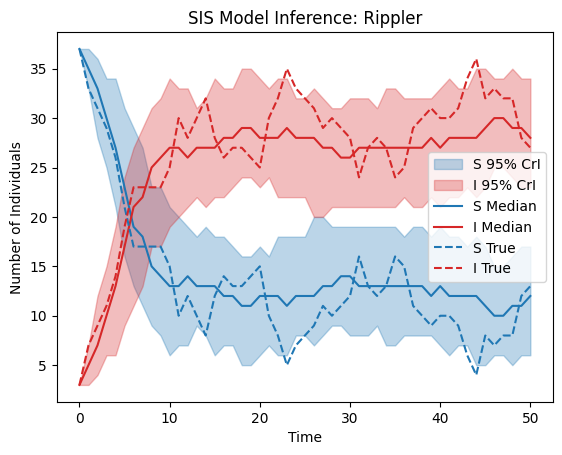}
     \end{subfigure}
     \hfill
     \begin{subfigure}[h]{0.32\textwidth}
         \centering
         \includegraphics[width=\textwidth]{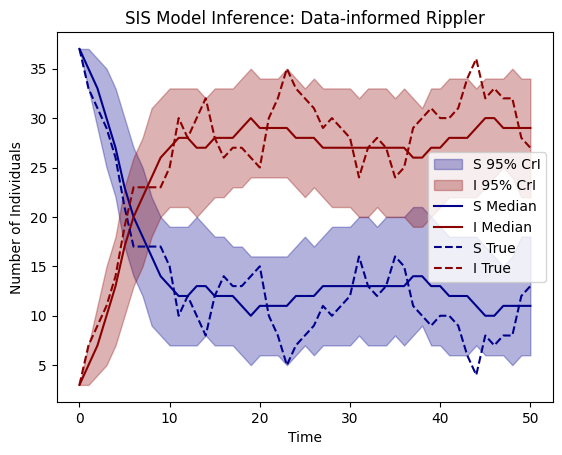}
     \end{subfigure}
     \hfill
     \begin{subfigure}[h]{0.32\textwidth}
         \centering
         \includegraphics[width=\textwidth]{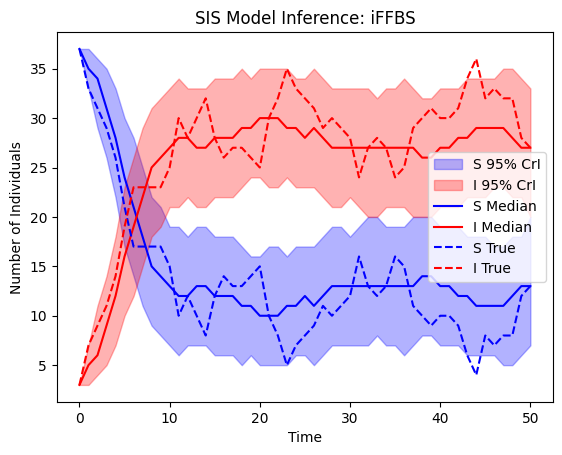}
     \end{subfigure}
     \caption{Posterior samples of the number of individuals in each state over time, using the Rippler (left), data-informed Rippler (middle), and iFFBS (right) methods on the multi-strain model ($\Sc_I=3$ strains). The median values are shown in the solid lines and the 95\% credible intervals are shown in the shaded regions. The true values are shown by the dashed lines.}
     \label{fig_sis_credible_intervals}
\end{figure}

\end{document}